%% file: cdel.tex
\documentclass[11pt]{article}
\usepackage{cdel}
\title{Quantum encryption with certified deletion}
\author{Anne Broadbent and Rabib Islam%
\footnote{University of Ottawa, Department of Mathematics and Statistics; \texttt{\{abroadbe,risla028\}@uottawa.ca}.}
}
\date{}
\begin{document}
\maketitle

\begin{abstract}
\input{tex/0_abstract.tex}
\end{abstract}

\input{tex/1_introduction.tex}

\input{tex/2_preliminaries.tex}

\input{tex/3_security_definitions.tex}

\input{tex/4_construction.tex}

\input{tex/5_security_analysis.tex}

\bibliographystyle{bib/alphaarxiv}
\bibliography{bib/full,bib/quantum,bib/eur}
\end{document}

%% file: tex/0_abstract.tex
Given a ciphertext, is it possible to prove the \emph{deletion} of the underlying plaintext?
Since classical ciphertexts can be copied, clearly such a feat is impossible using classical information alone.
In stark contrast to this, we show that quantum encodings enable \emph{certified deletion}.
More precisely, we show that it is possible to encrypt classical data into a quantum ciphertext such that the recipient of the ciphertext can produce a \emph{classical} string which proves to the originator that the recipient has relinquished any chance of recovering the plaintext should the decryption key be revealed.
Our scheme is feasible with current quantum technology: the honest parties only require quantum devices for single-qubit preparation and measurements; the scheme is also robust against noise in these devices. Furthermore, we provide an analysis that is suitable in the finite-key regime.

%% file: tex/1_introduction.tex
\section{Introduction}

Consider the following scenario: Alice sends a ciphertext to Bob, but in addition,
she wants to encode the data in a way such that Bob can prove to her that he \emph{deleted} the information contained in the ciphertext. Such a
deletion should prevent Bob from retrieving any information on the encoded plaintext once the key is revealed.
We call this \emph{certified deletion}.

Informally, this functionality stipulates that Bob should not be able to do the following two things simultaneously:
\begin{inparaenum}[(1)]
    \item Convince Alice that he has deleted the ciphertext; and
    \item Given the key, recover information about the encrypted message.
\end{inparaenum}
To better understand this concept, consider an analogy to certified deletion in the physical world: ``encryption''  would correspond to locking information into a keyed safe, the ``ciphertext'' comprising of the locked safe. In this case, ``deletion'' may simply involve  returning the safe in its original state.
This ``deletion'' is intrinsically certified since, without the safe (and having never had access to the key and the safe at the same time), Bob is relinquishing the possibility of gaining access to the information (even in the future when the key may be revealed) by returning the safe.
However, in the case that encryption is digital, Bob may retain a copy of the ciphertext; there is therefore no meaningful way for him to certify ``deletion'' of the underlying information, since clearly a copy of the ciphertext is just as good as the original ciphertext, when it comes time to use the key to decrypt the data.

Quantum information, on the other hand, is known for its no-cloning principle~\cite{Die82, Par70, WZ82}, which states that quantum states cannot, in general, be copied.
This quantum feature has been explored in many cryptographic applications, including unforgeable money~\cite{Wie83}, quantum key distribution (QKD)~\cite{BB84}, and more (for a survey, see~\cite{BS16}).

\subsection{Summary of Contributions}\label{sec:intro-summary-contributions}

In this work, we add to the repertoire of functionalities that are classically impossible but are achievable with unconditional security by means of quantum information.
We give a formal definition of certified deletion encryption and certified deletion security.
Moreover, we construct an encryption scheme which, as we demonstrate, satisfies these notions (in addition, our proofs are applicable in the finite-key regime).
Furthermore, our scheme is technologically simple since it can be implemented by honest parties who have access to rudimentary quantum devices (that is, they only need to prepare single-qubit quantum states, and perform single-qubit measurements); we also show that our scheme is robust against noise in these devices.
We now elaborate on these contributions.

\subsubsection{Definitions}\label{sec:intro-definitions}

In order to define our notion of encryption, we build on the \emph{quantum encryption of classical messages} (QECM) framework~\cite{BL20}\footnote{Apart from sharing this basic definition, our work differs significantly from~\cite{BL20}.
For instance, the adversarial models are fundamentally different, since we consider here a single adversary, while~\cite{BL20} is secure against \emph{two} separate adversaries.}
(for simplicity, our work is restricted to the single-use, private-key setting).
To the QECM, we add a \emph{delete} circuit which is used by Bob if he wishes to delete his ciphertext and generate a corresponding verification state, and a \emph{verify} circuit which uses the key and is used by Alice to determine whether Bob really deleted the ciphertext.

Next, we define the notion of certified deletion security for a QECM scheme (See \cref{fig:cer-del} and \cref{def:cds}).
Our definition is inspired by elements of the definition in~\cite{Unr14}.
The starting point for this definition is the well-known indistinguishability experiment, this time played between an adversary~$\lA=(\lA_0, \lA_1, \lA_2)$ and a challenger.
After running the Key Generation procedure, the adversary $\lA_0$ submits an $n$-bit plaintext~$\textsf{msg}_0$ to the challenger.
Depending on a random bit~$b$, the challenger either encrypts~$\textsf{msg}_0$ or a dummy plaintext~$0^n$, and sends the ciphertext to~$\lA_1$.
The adversary~$\lA_1$ then produces a candidate classical ``deletion certificate'', $y$. Next, the key is sent to the adversary~$\lA_2$ who produces an output bit~$b' \in \{0,1\}$.\footnote{The key is leaked \emph{after} $y$ is produced; this is required because otherwise, with access to the ciphertext and the key, the adversary could (via purification) retrieve the plaintext without affecting the ciphertext, and therefore could decrypt while simultaneously producing a convincing proof of deletion.}
A scheme is deemed \emph{secure} if the choice of $b$ does not change the probability of the following event: ``$b' = 1$ \emph {and} the deletion certificate $y$ is accepted''.
We note that it would be incorrect to formulate a definition that conditions on~$y$ being accepted (see discussion in~\cite{Unr14}).
We note that certified deletion security does not necessarily imply ciphertext indistinguishability; hence these two properties are defined and proven separately.

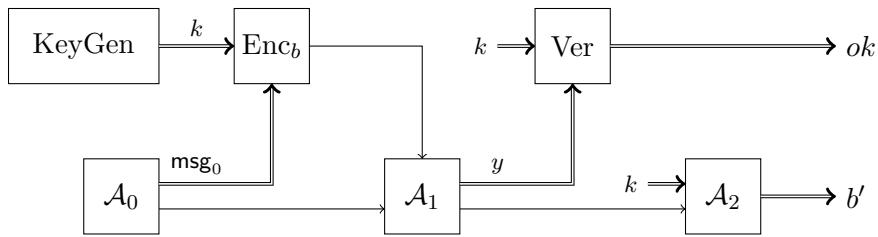
\begin{figure}[h]
    \centering
    \begin{tikzpicture}
        \draw (0,0) rectangle (2,-1);
        \node at (1,-0.5) {KeyGen};
        \draw[double, ->] (2,-1/2) -- (3,-1/2);
        \node[above] at (2.5,-1/2) {\footnotesize{$k$}};
        \draw (3,0) rectangle (4,-1);
        \node at (3.5,-0.5) {$\textrm{Enc}_b$};
        \draw[->] (4,-1/2) -- (5.5,-1/2) -- (5.5,-2);
        \draw[double, ->] (6.5,-1/2) -- (7,-1/2);
        \node[left] at (6.5,-1/2) {\footnotesize{$k$}};
        \draw (7,0) rectangle (8,-1);
        \node at (7.5,-1/2) {Ver};
        \draw[double, ->] (8,-0.5) -- (11,-0.5);
        \node[right] at (11,-0.5) {$ok$};
        \draw (1,-2) rectangle (2,-3);
        \node at (1.5,-2.5) {$\lA_0$};
        \draw[double, ->] (2,-7/3) -- (3.5,-7/3) -- (3.5,-1);
        \node[above] at (2.5,-7/3) {\footnotesize{$\msg_0$}};
        \draw[->] (2,-8/3) -- (5,-8/3);
        \draw (5,-2) rectangle (6,-3);
        \node at (5.5,-2.5) {$\lA_1$};
        \draw[double, ->] (6,-7/3) -- (7.5,-7/3) -- (7.5,-1);
        \node[above] at (6.5,-7/3) {\footnotesize$y$};
        \draw[->] (6,-8/3) -- (9,-8/3);
        \draw[double, ->] (8.5,-7/3) -- (9,-7/3);
        \node[left] at (8.5,-7/3) {\footnotesize{$k$}};
        \draw (9,-2) rectangle (10,-3);
        \node at (9.5,-2.5) {$\lA_2$};
        \draw[double, ->] (10,-2.5) -- (11,-2.5);
        \node[right] at (11,-2.5) {$b'$};
    \end{tikzpicture}

       \caption{Schematic representation of the security notion for certified deletion security.
       The game is parametrized by $b \in \{0,1\}$ and $\textrm{Enc}_0$ outputs an encryption of~$0^n$ while $\textrm{Enc}_1$ encrypts its input, $\msg_0$.
   Security holds if for each adversary $\lA=(\lA_0, \lA_1, \lA_2)$, the probability of ($b' = 1$ \emph{and} $ok = 1$) is essentially the same, regardless of the value of $b$.}
    \label{fig:cer-del}
\end{figure}

\subsubsection{Scheme}

In \cref{sec:pm}, we  present our scheme.
Our encoding is based on the well-known Wiesner encoding~\cite{Wie83}.
Informally, the message is encoded by first generating $m$ random Wiesner states, $\ket{r}^\theta$ ($r, \theta \in \{0,1\}^m$) (for notation, see \cref{ss:notation}).
We let $r|_\lI$  be the substring of~$r$ where qubits are encoded in the computational basis, and we let $r|_{\bar{\lI}}$ be the remaining substring of~$r$ (where qubits are encoded in the Hadamard basis).
Then, in order to create a classical \emph{proof of deletion}, Bob measures the entire ciphertext in the Hadamard basis.
The result is a classical string, and Alice accepts the deletion if all the bits corresponding to positions encoded in the Hadamard basis are correct according to~$r|_{\bar{\lI}}$.
As for the message~$\msg$, it is encoded into $x'= \msg \oplus H(r|_\lI) \oplus u$, where $H$ is a two-universal hash function and $u$ is a fresh, random string.
Intuitively speaking, the  use of the hash function is required in order to prevent that \emph{partial} information retained by Bob could be useful in distinguishing the plaintext, while the random $u$ is used to guarantee security in terms of an encryption scheme.
Robustness of the protocol is achieved by using an error correcting code and including an encrypted version of the error syndrome.
We note that while our definitions do not require it, our scheme provides a further desirable property, namely that the proof of deletion is a classical string only.

\subsubsection{Proof}

In \cref{sec:sec}, we present the security analysis of our scheme and give concrete security parameters (\cref{thm:cert-del} and its proof).
First, the fact that the scheme is an encryption scheme is relatively straightforward; it follows via a generalization of the quantum one-time pad (see \cref{sec:indi-security}).
Next, correctness and robustness (\cref{sec:correctness}) follow from the properties of the encoding and of the error correcting mechanism.

Next, the proof of security for certified deletion has a number of key steps.
First, we apply the security notion of certified deletion~(\cref{def:cds}) to our concrete scheme (\cref{sch:pmcd}).
This yields a ``prepare-and-measure'' security game (see \cref{game:pm}).
However, for the purposes of the analysis, it is convenient to consider instead an entanglement-based game (this is a common proof technique for quantum protocols that include the preparation of random states \cite{LC99,SP00}).
In this game (\cref{game:epr}), the adversary, Bob, creates an initial entangled state, from which Alice derives (via measurements in a random basis $\theta$ of her choosing) the value of $r \in \{0,1\}^m$.
We show that, without loss of generality, Bob can produce the proof of deletion, $y$, \emph{before} he receives any information from Alice (this is due, essentially, to the fact that the ciphertext is uniformly random from Bob's point of view).
Averaging over Alice's choice of basis $\theta$, we arrive at a very powerful intuition:
in order for Bob's probability of creating an acceptable proof of deletion $y$ (\emph{i.e.} he produces a string where the positions corresponding to $\theta=1$ match with $ r|_{\bar{\lI}}$) to be high, he must unavoidably have a low probability of correctly guessing $r|_\lI$.
The above phenomenon is embodied in the following
\emph{entropic uncertainty relation} for smooth entropies~\cite{TR11,TLGR12}.
We consider the scenario of Eve preparing a tripartite state $\rho_{ABE}$ with Alice, Bob, and Eve receiving the $A$, $B$ and $E$ systems, respectively (here, $A$ and $B$ contain~$n$ qubits).
Next, Alice either measures all of her qubits in the computational basis to obtain string~$X$, or she measures all of her qubits in the Hadamard basis to obtain string~$Z$; meanwhile, Bob measures his qubits in the Hadamard basis to obtain~$Z'$.
We then have the relation:
   \begin{equation}
        H_{\min} ^\epsilon (X \mid E) + H_{\max} ^\epsilon (Z \mid Z') \geq n,
    \end{equation}
In the above, $\epsilon \leq 0$ is a smoothing parameter which represents a probability of failure, and the smooth min-entropy $H_{\min} ^\epsilon(X\mid E)$ characterizes the average probability that Eve guesses~$X$ correctly using her optimal strategy and given her quantum register~$E$, while the smooth max-entropy $H_{\max} ^\epsilon(Z\mid Z')$ corresponds to the number of bits that are needed in order to reconstruct $Z$ from $Z'$ up to a failure probability~$\epsilon$ (for details, see \cref{sec:entropic-uncertainty-rels}).

Our proof technique thus consists in formally analysing the entanglement-based game and applying the appropriate uncertainty relation in the spirit of the one above.
Finally, we combine the bound on Bob's min-entropy with a universal$_2$ hash function and the Leftover Hashing Lemma of~\cite{Ren05} to prove indistinguishability between the cases $b=0$ and $b=1$ after Alice has been convinced of deletion.

\subsection{Related Work}
To the best of our knowledge, the first use of a quantum encoding to certify that a ciphertext is completely ``returned'' was  developed by Unruh~\cite{Unr14} in the context of \emph{revocable timed-release encryption}\footnote{Revocable timed-release encryption can be equivalently thought of as a revocable time-lock puzzle\cite{RSW96}, which does not satisfy standard cryptographic security (since the plaintexts are recoverable, by design, in polynomial time). In contrast, here we achieve a semantic-security-type security definition.}: in this case, the revocation process is fully quantum.
Our main security definition (\cref{def:cds}) is closely related to the security definitions from this work.
On the technical side, our work differs significantly since~\cite{Unr14} uses techniques related to CSS codes and quantum random oracles, whereas we use privacy amplification and uncertainty relations.
Our work also considers the concept of ``revocation'' outside the context of timed-release encryption, and it is also a conceptual and technical improvement since it shows that a proof of deletion can be classical.
Fu and Miller~\cite{FM18} gave the first evidence that quantum information could be used to prove \emph{deletion} of information and that this could be verified using classical interaction only: they showed that, via a two-party nonlocality game (involving classical interaction), Alice can become convinced that Bob has \emph{deleted} a single-bit ciphertext (in the sense that the deleted state is unreadable even if Bob were to learn the decryption  key).
Their results are cast in the device-independent setting (meaning that security holds against arbitrarily malicious quantum devices).
Further related work (that is independent from ours) by Coiteux-Roy and Wolf~\cite{CW19} touches on the question of provable deletion using quantum encodings.
However, their work is not concerned with encryption schemes, and therefore does not consider leaking of the key.
By contrast, we are explicitly concerned with what it would mean to delete a quantum ciphertext.
We note, however, that there are similarities between our scheme and the proposed scheme in~\cite{CW19}, namely the use of conjugate coding, with the message encoded in one basis and the conjugate basis, to prove deletion.

\paragraph{Relationship with Quantum Key Distribution.} It can be instructive to compare our results to the ones obtained in the analysis of QKD~\cite{TL17}.
Firstly, our adversarial model appears different since in certified deletion, we have one honest party (Alice, the sender) and one cheating party (Bob, the receiver), whereas QKD involves two honest parties (Alice and Bob) and one adversary (Eve).
Next, the interaction model is different since certified deletion is almost non-interactive, whereas QKD involves various rounds of interaction between Alice and Bob.
However, the procedures and proof techniques for certified deletion are close to the ones used in QKD: we use similar encodings into Wiesner states, similar privacy amplification and error correction, and the analysis via an entanglement-based game uses similar entropic uncertainty relations, leading to a security parameter that is very similar to the one in \cite{TL17}.
While we are not aware of any direct reduction from the security of a QKD scheme to certified deletion, we note that, as part of our proof technique, we manage to essentially map the adversarial model for certified deletion to one similar to the  QKD model since we \emph{split} the behaviour of our adversarial Bob into multiple phases: preparation of the joint state $\rho_{ABE}$, measurement of a register $B$ in a determined basis, and finally bounding the advantage that the adversary has in \emph{simultaneously} making Alice accept the outcome of the measurement performed on $B$ \emph{and} predicting some measurement outcome on register~$A$ given quantum side-information~$E$.
This scenario is similar to QKD, although we note that the measurement bases are not chosen randomly but are instead consistently in the Hadamard basis (for Bob's measurement) and that Eve's challenge is to predict Alice's measurement in the computational basis only (this situation is reminiscent of the \emph{single-basis parameter estimation} technique~\cite{TL17,PLWC16}).

\subsection{Applications and Open Questions}

While the main focus of this work is on the foundations of certified deletion, we can nevertheless envisage potential applications which we briefly discuss below (we leave the formal analyses for future work).

\paragraph*{Protection against data retention.} In 2016, the European Union adopted a regulation on the processing and free movement of personal data~\cite{EUR-2016-679}.
Included is a clause on the ``right to be forgotten'': a person should be able to have their data erased whenever its retention is no longer necessary.
See also~\cite{GGV20}.
Certified deletion encryption might help facilitate this scenario in the following way: if a party were to provide their data to an organization via a certified deletion encryption, the organization would be able to certify deletion of the data using the deletion circuit included in the~scheme.
Future work could develop a type of homomorphic encryption with certified deletion so that the ciphertexts could be useful to some extent while a level of security, in terms of deletion, is maintained.
Also useful would be a type of ``public verifiability'' which would enable parties other than the originator to verify deletion certificates. Contact tracing~\cite{CTV20arxiv} is another relevant scenario where individual data could be safeguarded against data retention by using certified deletion.

\paragraph*{Encryption with classical revocation.}

The concept of \emph{ciphertext revocation} allows a recipient to provably \emph{return} a ciphertext (in the sense that the sender can confirm that the ciphertext is returned and that the recipient will \emph{not} be able to decrypt, even if the key is leaked in the future); such a functionality is unachievable with classical information alone, but it is known to be achievable using quantum ciphertexts~\cite{Unr14}.
In a sense, our contribution is an extension of revocation since from the point of view of the recipient, whether quantum information is deleted or returned, the end result is similar: the recipient is unable to decrypt even given the decryption key.
Our scheme, however, has the advantage of using classical information only for the deletion.

As a use case for classical revocation, consider a situation where Bob loans Alice an amount of money. Alice agrees to pay back the full amount in time~$T$ plus 15 percent interest if Bob does not recall the loan within that time.
To implement this scheme, Alice uses a certified deletion encryption scheme to send Bob an encrypted cheque and schedules her computer to send Bob the key at time~$T$.
If Bob wishes to recall the loan within time $T$, he sends Alice the deletion string. Another possible application is \emph{timed-release encryption}~\cite{Unr14}, where the key is included in the ciphertext, but with the ciphertext encoded in a classical timed-release encryption.

\paragraph*{Composable and Everlasting Security.}

We leave as an open question the composability of our scheme (as well as security beyond the one-time case).
We note that through a combination of composability with our quantum encoding, it may be possible to transform a long-term computational assumption into a temporary one.
That is, a computational assumption would need to be broken \emph{during} a protocol, or else the security would be information-theoretically secure as soon as the protocol ends.
This is called \emph{everlasting security}~\cite{Unr13}.

For example, consider the situation encountered in a zero-knowledge proof system for a $\Sigma$-protocol (for instance, for graph 3-colouring \cite{GMW91}): the prover commits to an encoding of an $\textsf{NP}$-witness using a statistically binding and computationally concealing commitment scheme. The verifier then randomly chooses which commitments to open, and the prover provides the information required to open the commitment.
If, in addition, we could encode the commitments with a scheme that provides composable certified deletion, then the verifier could also prove that the unopened commitments are effectively \emph{deleted}.
This has the potential of ensuring that the zero-knowledge property becomes \emph{statistical} as long as the computational assumption is not broken \emph{during} the execution of the proof system.
This description assumes an extension of our certified deletion encoding to the computational setting and also somehow assumes that the verifier would collaborate in its deletion actions (we leave for future work the formal statement and analysis). Nevertheless, since zero-knowledge proofs are building blocks for a host of cryptographic protocols, certified deletion has the potential to unleash everlasting security; this is highly desirable given steady progress in both algorithms and quantum computers. Another potential application would be proving erasure (in the context where there is no encryption)~\cite{CW19}.

\subsection{Outline}

The remainder of this paper is structured as follows.
\cref{sec:prelim} is an introduction to concepts and notation used in the rest of this work.
\cref{sec:definitions} lays out the novel security definitions which appear in this paper.
\cref{sec:pm} is an exposition of our main scheme, while \cref{sec:sec} provides a security analysis.

\subsection*{Acknowledgements}
We would like to thank Carl Miller for related discussions.
We are grateful to the anonymous reviewers for pointing out a mistake in a previous version of the security definition.
This work was supported by the U.S. Air Force Office of
Scientific Research under award number FA9550-17-1-0083, Canada's NSERC, an Ontario ERA, and the University of Ottawa's Research Chairs program.

%% file: tex/2_preliminaries.tex
\section{Preliminaries}\label{sec:prelim}

In this section, we outline certain concepts and notational conventions which are used throughout the article.
We assume that the reader has a basic familiarity with quantum computation and quantum information.
We refer to~\cite{NC00} for further background.

\subsection{Notation}\label{ss:notation}

We make use of the following notation: for a function $f \colon X \rightarrow \bR$, we denote
\begin{equation}
    \mathop{\mathbb{E}}_x f(x) = \frac{1}{|X|} \sum_{x \in X} f(x).
\end{equation}

We represent the Hamming weight of strings as the output of a Hamming weight function $\omega \colon \{0,1\}^* \rightarrow \bN$.
If $x_1, \ldots, x_n$ are strings, then we define $(x_1, \ldots, x_n)$ to be the concatenation of these strings.
Let $[n]$ denote the set $\{ 1, 2, \ldots, n \}$.
Then, for any string $x = (x_1, \ldots, x_n)$ and any subset $\lI \subseteq [n]$, we use $x |_\lI$ to denote the string $x$ restricted to the bits indexed by $\lI$.
We call a function $\eta \colon \bN \rightarrow \bR_{\geq 0}$ \emph{negligible} if for every positive polynomial $p$, there exists an integer $N$ such that, for all integers $n > N$, it is true that $\eta(n) < \frac{1}{p(n)}$.

We let $\lQ \coloneqq \bC^2$ denote the state space of a single qubit, and we use the notation $\lQ(n) \coloneqq \lQ^{\otimes n}$ for any $n \in \bN$.
Let $\lH$ be a Hilbert space.
The group of unitary operators on $\lH$ is denoted by $\lU(\lH)$, and the set of density operators on $\lH$ is denoted by $\lD(\lH)$.
Through density operators, a Hilbert space may correspond to a \emph{quantum system}, which we represent by capital letters.
The set of diagonal density operators on $\lH$ is denoted by $\fD(\lH)$---the elements of this set represent classical states.
Discrete random variables are thus modeled as finite-dimensional quantum systems, called \emph{registers}.
A register $X$ takes values in $\lX$.
A density operator $\egoketbra{x}$ will be denoted as $\ketbra{x}$.
We employ the operator norm, which we define for a linear operator $A \colon \lH \rightarrow \lH'$ between finite-dimensional Hilbert spaces $\lH$ and $\lH'$ as
\begin{equation}
    \| A \| = \sup \{ \| Av \| \mid v \in \lH, \| v \| = 1 \}.
\end{equation}
Moreover, for two density operators $\rho, \sigma \in \lD(\lH)$, we use the notation $\rho \leq \sigma$ to say that $\sigma - \rho$ is positive semi-definite.

In order to illustrate correlations between a classical register $X$ and a quantum state $A$, we use the formalism of a \emph{classical-quantum} state:
\begin{equation}
    \rho_{XA} = \sum_{x \in \lX} P_X (x) \ketbra{x}_X \otimes \rho_{A \mid X = x},
\end{equation}
where $P_X (x) \coloneqq \Pr[X = x]_\rho = \Tr[\ketbra{x}_X \rho_{XA}]$ and $\rho_{A \mid X = x}$ is the state of $A$ conditioned on the event that $X = x$.

Let $\ketbra{x_i} \in \fD(\lH)$ be classical states for integers $i$ such that $1 \leq i \leq n$.
Then we use the notation
\begin{equation}
    \ketbra{x_1, x_2, \ldots, x_n} \coloneqq \ketbra{x_1} \otimes \ketbra{x_2} \otimes \cdots \otimes \ketbra{x_n}.
\end{equation}

Let $\sH \in \lU(\lQ)$ denote the Hadamard operator, which is defined by
\begin{equation}
    \ket{0} \mapsto \frac{\ket{0} + \ket{1}}{\sqrt{2}}, \quad \ket{1} \mapsto \frac{\ket{0} - \ket{1}}{\sqrt{2}}.
\end{equation}
For any strings $x, \theta \in \{0,1\}^n$, we define
\begin{equation}
    \ket{x^\theta} = \sH^\theta \ket{x} = \sH^{\theta_1} \ket{x_1} \otimes \sH^{\theta_2} \ket{x_2} \otimes \cdots \otimes \sH^{\theta_n} \ket{x_n}.
\end{equation}
States of the form $\ket{x^\theta}$ are here called \emph{Wiesner states} in recognition of their first use in~\cite{Wie83}.

We make use of the Einstein-Podolsky-Rosen (EPR) state~\cite{EPR35}, defined as
\begin{equation}
    \ket{\EPR} = \frac{1}{\sqrt{2}} (\ket{0} \otimes \ket{0} + \ket{1} \otimes \ket{1}).
\end{equation}

We use $x \unif X$ to denote sampling an element $x \in X$ uniformly at random from a set $X$.
This uniform randomness is represented in terms of registers in the fully mixed state which is, given a $d$-dimensional Hilbert space $\lH$, defined as
    $\frac{1}{d} 1_{d}$,
where $1_d$ denotes the identity matrix with $d$ rows.

For two quantum states $\rho, \sigma \in \lD(\lH)$, we define the \emph{trace distance}
\begin{equation}
    \lVert \rho - \sigma \rVert_{\Tr} \coloneqq \frac{1}{2} \lVert \rho - \sigma \rVert.
\end{equation}
Note also an alternative formula for the trace distance:
\begin{equation}
    \lVert \rho - \sigma \rVert_{\Tr} = \max_P \Tr[P(\rho - \sigma)],
\end{equation}
where $P \leq 1_d$ is a positive operator. Hence, in terms of a physical interpretation, the trace distance is the upper bound for the difference in probabilities with respect to the states $\rho$ and $\sigma$ that a measurement outcome $P$ may occur on the state.

We define purified distance, which is a metric on quantum states.
\begin{definition}[Purified Distance]
    Let $A$ be a quantum system.
    For two (subnormalized) states $\rho_A, \sigma_A$, we define the \emph{generalized fidelity},
    \begin{equation}
        F(\rho_A, \sigma_A) \coloneqq \left( \Tr \left[ \sqrt{\sqrt{\rho_A} \sigma_A \sqrt{\rho_A}} \right] + \sqrt{1 - \Tr[\rho_A]} \sqrt{1 - \Tr[\sigma_A]} \right)^2,
    \end{equation}
    and the \emph{purified distance},
    \begin{equation}
        P(\rho_A, \sigma_A) \coloneqq \sqrt{1 - F(\rho_A, \sigma_A)}.
    \end{equation}
\end{definition}

\subsection{Hash Functions and Error Correction}

We make use of universal$_2$ hash functions, first introduced by Carter and Wegman~\cite{CW79}.
\begin{definition}[Universal$_2$ Hashing]
    Let $\fH = \{H \colon \lX \rightarrow \lZ\}$ be a family of functions.
    We say that $\fH$ is \emph{universal}$_2$ if $\Pr[H(x) = H(x')] \leq \frac{1}{| \lZ |}$ for any two distinct elements $x, x' \in \lX$, when $H$ is chosen uniformly at random from $\fH$.
\end{definition}
Such families exist if $|\lZ|$ is a power of two (see~\cite{CW79}).
Moreover, there exist universal$_2$ families of hash functions which take strings of length $n$ as input and which contain $2^{O(n)}$ hash functions; therefore it takes $O(n)$ bits to specify a hash function from such a family~\cite{WC81}.
Thus, when we discuss communication of hash functions, we assume that both the sender and the recipient are aware of the family from which a hash function has been chosen, and that the transmitted data consists of $O(n)$ bits used to specify the hash function from the known family.

In the context of error correction, we note that linear error correcting codes can generate syndromes, and that corrections to a message can be made when given the syndrome of the correct message. This is called syndrome decoding.
Therefore, we implicitly refer to syndrome decoding of an $[n, n-s]$-linear code which handles codewords of length $n$ and generates syndromes of length $s < n$ when we use functions $\synd \colon \{0,1\}^n \rightarrow \{0,1\}^s$ and $\corr \colon \{0,1\}^n \times \{0,1\}^s \rightarrow \{0,1\}^n$, where $\synd$ is a syndrome-generating function and $\corr$ is a string-correcting function.
We also make reference to the distance of an error correcting code, which is the minimum distance between distinct codewords.

\subsection{Quantum Channels and Measurements}\label{ss:prelim-msmt}

Let $A$ and $B$ be two quantum systems, and let $X$ be a classical register.
A \emph{quantum channel} $\Phi \colon A \rightarrow B$ is a completely positive trace-preserving (CPTP) map.
A \emph{generalized measurement} on $A$ is a set of linear operators $\{M^x _A\}_{x \in \lX}$, where $x \in \lX$ are potential classical outcomes, such that
\begin{equation}
    \sum_{x \in \lX} (M_A ^x )^\dagger (M_A ^x) = 1_A.
\end{equation}
A \emph{positive-operator valued measure} (POVM) on $A$ is a set of Hermitian positive semidefinite operators $\{M^x _A\}_{x \in \lX}$, where $x \in \lX$ are potential classical outcomes, such that
\begin{equation}
    \sum_{x \in \lX} (M_A ^x) = 1_A.
\end{equation}
We also represent measurements with CPTP maps such as $\lM_{A \rightarrow X}$, which map quantum states in system $A$ to classical states in register $X$ using POVMs.

For two registers $X$ and $Y$, if we have a function, $f \colon \lX \rightarrow \lY$ then we denote by $\lE_f \colon X \rightarrow XY$ the CPTP map
\begin{equation}
    \lE_f [\cdot] \coloneqq \sum_{x \in X} \ket{f(x)}_Y \ketbra{x}_X \cdot \ketbra{x}_X \bra{f(x)}_Y.
\end{equation}

In this work, measurement of a qubit in our scheme will always occur in one of two bases: the computational basis ($\{ \ket{0}, \ket{1} \}$) or the Hadamard basis $(\{ \ket{+}, \ket{-} \})$.
Thus, for a quantum system $A$, we notate these measurements as $ \{ M_{A} ^{\theta, x} \}_{x \in \{0,1\}}$, where $x \in \{0,1\}$ ranges over the possible outcomes, and where $\theta \in \{0,1\}$ determines the basis of measurement ($\theta = 0$ indicates computational basis and $\theta = 1$ indicates Hadamard basis).

Let $\{M^x _A\}_x$ and $\{N^y _A\}_y$ be two POVMs acting on a quantum system $A$.
We define the overlap
\begin{equation}
    c(\{M^x _A\}_x, \{N^y _B\}_y) \coloneqq \max_{x,y} \left \lVert \sqrt{M^x _A} \sqrt{N^y _A} \right \rVert ^2 _\infty .
\end{equation}
Wherever dealing with an $m$-qubit quantum system $A$, we define, for all $i = 1, \ldots, m$,
\begin{equation}
    c_i \coloneqq c \left( \{M^{0,x} _{A_i}\}_x , \{M^{1,y} _{A_i}\}_y \right).
\end{equation}
We assume our measurements are ideal, so $c_i = 1/2$.

\subsection{Entropic Uncertainty Relations}\label{sec:entropic-uncertainty-rels}

The purpose of entropy is to quantify the amount of uncertainty an observer has concerning the outcome of a random variable.
Since the uncertainty of random variables can be understood in different ways, there exist different kinds of entropy.
Key to our work are min- and max-entropy, first introduced by Renner and K{\"o}nig~\cite{Ren05, KRS09}, as a generalization of conditional R{\'e}nyi entropies~\cite{Ren61} to the quantum setting.
Min-entropy, for instance, quantifies the degree of uniformity of the distribution of a random variable.
\begin{definition}[Min-entropy]
    Let $A$ and $B$ be two quantum systems.
    For any bipartite state $\rho_{AB}$, we define
    \begin{equation}
        H_{\min} (A \mid B)_\rho \coloneqq \sup \{ \xi \in \bR \mid \exists~\textrm{state}~\sigma_B~\textrm{such that}~\rho_{AB} \leq 2^{-\xi} 1_A \otimes \sigma_B \}.
    \end{equation}
\end{definition}

Max-entropy quantifies the size of the support of a random variable, and is here defined by its dual relation to min-entropy.
\begin{definition}[Max-entropy]
    Let $A$ and $B$ be two quantum systems.
    For any bipartite state $\rho_{AB}$, we define
    \begin{equation}
        H_{\max} (A \mid B)_\rho \coloneqq - H_{\min} (A \mid C)_\rho,
    \end{equation}
    where $\rho_{ABC}$ is any pure state with $\Tr_C [\rho_{ABC}] = \rho_{AB}$, for some quantum system~$C$.
\end{definition}

In order to deal with finite-size effects, it is necessary to generalize min- and max-entropy to their smooth variants.
\begin{definition}[Smooth Entropies] \label{def:smoothentropies}
    Let $A$ and $B$ be two quantum systems.
    For any bipartite state $\rho_{AB}$, and $\epsilon \in \left[ 0, \sqrt{\Tr[\rho_{AB}]} \right)$, we define
    \begin{align}
        H_{\min}^\epsilon (A \mid B)_\rho \coloneqq \sup_{\substack{\tilde \rho_{AB} \\ P(\tilde \rho_{AB}, \rho_{AB}) \leq \epsilon}} H_{\min} (A \mid B)_{\tilde \rho},\\
        H_{\max}^\epsilon (A \mid B)_\rho \coloneqq \inf_{\substack{\tilde \rho_{AB} \\ P(\tilde \rho_{AB}, \rho_{AB}) \leq \epsilon}} H_{\max} (A \mid B)_{\tilde \rho}.
    \end{align}
\end{definition}

It is of note that smooth entropies satisfy the following inequality, commonly referred to as the data-processing inequality~\cite{TCR10}.
\begin{proposition}
    \label{prop:data}
    Let $\epsilon \geq 0$, $\rho_{AB}$ be a quantum state, and $\lE \colon \lD(\lH_A) \rightarrow \lD(\lH_C)$ be a CPTP map.
Define $\sigma_{AC} \coloneqq ( 1_{\lD(\lH_A)} \otimes \lE ) (\rho_{AB})$.
Then,
    \begin{equation}
        H_{\min} ^\epsilon (A \mid B)_\rho  \leq H_{\min} ^\epsilon (A \mid C)_\sigma \quad \textrm{and} \quad H_{\max} ^\epsilon (A \mid B )_\rho \leq H_{\max} ^\epsilon (A \mid C)_\sigma.
    \end{equation}
\end{proposition}

We use one half of the generalized uncertainty relation theorem found in~\cite{Tom12}, the precursor of which was introduced by Tomamichel and Renner~\cite{TR11}.
The original uncertainty relation was understood in terms of its application to QKD, and was used to prove the secrecy of the key in a finite-key analysis of QKD~\cite{TLGR12}.
\begin{proposition}
    \label{prop:unc}
    Let $\epsilon \geq 0$, let $\rho_{ACE}$ be a tripartite quantum state and let $\{M^x _A \}_{x \in \lX}$ and $\{N^z _A \}_{z \in \lZ}$ be two POVMs acting on $A$, and let $\{P^k _A\}_{k \in \lK}$ be a projective measurement acting on $A$.
    Then the post-measurement states
    \begin{equation}
        \rho_{XKC} = \sum_{x, k} \braket{x} \otimes \braket{k} \otimes \Tr_{AE} \left[ \sqrt{M^x _A} P^k _A \rho_{ACE} P^k _A \sqrt{M^x _A} \right]
    \end{equation}
    and
    \begin{equation}
    \rho_{YKE} = \sum_{y, k} \braket{y} \otimes \braket{k} \otimes \Tr_{AC} \left[ \sqrt{N^y _A} P^k _A \rho_{ACE} P^k _A \sqrt{N^y _A} \right]
    \end{equation}
    satisfy
    \begin{equation}
        H^\epsilon _{\min} (X \mid KC)_\rho + H^\epsilon _{\max} (Y \mid KE)_\rho \geq \log \frac{1}{c_\lK}
    \end{equation}
    where $c_\lK = \max_{k,x,y} \left \lVert \sqrt{M^x _A} P^k \sqrt{N^y _A} \right \rVert _\infty$.
\end{proposition}

We also use the Leftover Hashing Lemma, introduced by Renner~\cite{Ren05}.
It is typically understood in relation to the privacy amplification step of QKD.
We state it in the form given in~\cite{TL17}.

\begin{proposition}
\label{prop:lhl}
    Let $\epsilon \geq 0$ and $\sigma_{AX}$ be a classical-quantum state, with $X$ a classical register which takes values on $\lX = \{0,1\}^s$.
    Let $\fH$ be a universal$_2$ family of hash functions from $\lX$ to $\lY = \{0,1\}^n$.
    Let $\chi_Y = \frac{1}{2^n} 1_{\lD(\lY)}$ be the fully mixed state, $\rho_{S^H} = \frac{1}{| \fH |} \sum_{H \in \fH} \ketbra{H}_{S^H}$ and $\zeta_{AYS^H} = \Tr_X[\lE_f (\sigma_{AX} \otimes \rho_{S^H})]$ for the function $f \colon (x,H) \mapsto H(x)$ be the post-hashing state.
    Then,
    \begin{equation}
        \| \zeta_{A Y S^H} - \chi_{Y} \otimes \zeta_{A S^H} \|_{\Tr} \leq \frac{1}{2}2^{- \frac{1}{2} (H_{\min} ^\epsilon (X \mid A)_\sigma - n)} + 2 \epsilon.
    \end{equation}
\end{proposition}

\subsection{Statistical Lemmas}

The following lemmas are required to bound a specific max-entropy quantity.
They are both proven in~\cite{TL17} as part of a security proof of finite-key QKD, and this line of thinking originated in~\cite{TLGR12}.

The following lemma is a consequence of Serfling's bound~\cite{Ser74}.

\begin{lemma}
\label{lem:stat1}
    Let $Z_1, \ldots Z_m$ be random variables taking values in $\{0,1\}$.
    Let $m = s + k$.
    Let $\lI$ be an independent and uniformly chosen subset of $[m]$ with $s$ elements.
    Then, for $\nu \in [0,1]$ and $\delta \in (0,1),$
    \begin{equation}
        \Pr \left[\sum_{i \in \lI} Z_i \leq k \delta \wedge \sum_{i \in \bar \lI} Z_i \geq s (\delta + \nu) \right] \leq \exp \left( \frac{-2 \nu^2 s k^2}{m (k+1)} \right).
    \end{equation}
\end{lemma}

It will also be useful to condition a quantum state on future events.
The following lemma from~\cite{TL17} states that, given a classical-quantum state, there may exist a nearby state on which a certain event does not occur.

\begin{lemma}
\label{lem:stat2}
    Let $\rho_{AX}$ be a classical-quantum state with $X$ a classical register, and $\Omega \colon \lX \rightarrow \{0,1\}$ be an event with $\Pr[\Omega]_\rho = \epsilon < \Tr[\rho_{AX}]$.
    Then there exists a classical-quantum state $\tilde{\rho}_{AX}$ with $\Pr[\Omega]_{\tilde \rho} = 0$ and $P(\rho_{AX}, \tilde {\rho}_{AX}) \leq \sqrt{\epsilon}$.
\end{lemma}

\subsection{Quantum Encryption and Security}

Whenever an adversary $\lA$ is mentioned, it is assumed to be quantum and to have unbounded computational power, and we allow it to perform generalized measurements.

Considering that the scheme introduced in this paper is an encryption scheme with a quantum ciphertext, we rely on the ``quantum encryption of classical messages'' framework developed by Broadbent and Lord~\cite{BL20}.
This framework describes an encryption scheme as a set of parameterized CPTP maps which satisfy certain conditions.
\begin{definition}[Quantum Encryption of Classical Messages] \label{def:qecm}
    Let $n$ be an integer.
    An \emph{$n$-quantum encryption of classical messages} ($n$-QECM) is a tuple of uniform efficient quantum circuits $\lS = (\key, \enc, \dec)$ implementing CPTP maps of the form
    \begin{itemize}
        \item $\Phi_\lambda ^\key \colon \lD (\bC) \rightarrow \lD (\lH_{K, \lambda})$,
        \item $\Phi_\lambda ^\enc \colon \lD (\lH_{K, \lambda} \otimes \lH_{M}) \rightarrow \lD (\lH_{T, \lambda})$, and
        \item $\Phi_\lambda ^\dec \colon \lD (\lH_{K, \lambda} \otimes \lH_{T, \lambda}) \rightarrow \lD(\lH_{M})$,
    \end{itemize}
    where $\lH_{M} = \lQ(n)$ is the plaintext space, $\lH_{T, \lambda} = \lQ(\ell(\lambda))$ is the ciphertext space, and $\lH_{K, \lambda} = \lQ(\kappa(\lambda))$ is the key space for functions $\ell, \kappa \colon \bN^+ \rightarrow \bN^+$.

    For all $\lambda \in \bN^+, k \in \{0,1\}^{\kappa(\lambda)}$, and $m \in \{0,1\}^n$, the maps must satisfy
    \begin{equation}
        \Tr[\ketbra{k} \Phi^\key(1)] > 0 \Rightarrow \Tr[ \ketbra{m} \Phi^\dec _k \circ \Phi^\enc _k \ketbra{m}] = 1,
    \end{equation}
    where $\lambda$ is implicit, $\Phi_k ^\enc$ is the CPTP map defined by $\rho \mapsto \Phi^\enc (\ketbra{k} \otimes \rho )$, and we define $\Phi^\dec _k$ analogously.
    We also define the CPTP map $\Phi^\enc _{k,0} \colon \lD(\lH_{M}) \rightarrow \lD(\lH_{T, \lambda})$ by
    \begin{equation} \label{eq:qecm}
        \rho \mapsto \Phi^\enc _{k} (\ketbra{\boldsymbol{0}})
    \end{equation}
    where $\boldsymbol{0} \in \{0,1\}^n$ is the all-zero bit string, and the CPTP map $\Phi^\enc _{k,1} \colon \lD(\lH_{M}) \rightarrow \lD(\lH_{T, \lambda})$ by
    \begin{equation}
        \rho \mapsto \sum_{m \in \{0, 1\}^n} \Tr[\ketbra{m} \rho] \cdot \Phi^\enc _k (\ketbra{m}).
    \end{equation}
\end{definition}

As part of the security of our scheme, we wish to ensure that should an adversary obtain a copy of the ciphertext and were to know that the original message is one of two hypotheses, she would not be able to distinguish between the hypotheses.
We refer to this notion of security as ciphertext indistinguishability (called indistinguishable security in~\cite{BL20}).
It is best understood in terms of a scheme's resilience to an adversary performing what we refer to as a distinguishing attack.
\begin{definition}[Distinguishing Attack]
    Let $\lS = (\key, \enc, \dec)$ be an $n$-QECM.
    A \emph{distinguishing attack} is a quantum adversary $\lA = (\lA_0, \lA_1)$ implementing CPTP maps of the form
    \begin{itemize}
        \item $A_{0, \lambda} \colon \lD(\bC) \rightarrow \lD(\lH_{M} \otimes \lH_{S, \lambda})$ and
        \item $A_{1, \lambda} \colon \lD(\lH_{T, \lambda} \otimes \lH_{S, \lambda}) \rightarrow \lD(\lQ)$
    \end{itemize}
    where $\lH_{S, \lambda} = \lQ(s(\lambda))$ for a function $s \colon \bN^+ \rightarrow \bN^+$.
\end{definition}

\begin{definition}[Ciphertext Indistinguishability] \label{def:indistinguishable}
    Let $\lS = (\key, \enc, \dec)$ be an $n$-QECM.
    Then we say that $\lS$ has \emph{ciphertext indistinguishability} if for all distinguishing attacks $\lA$ there exists a negligible function~$\eta$ such that
    \begin{equation}
        \mathop{\mathbb{E}}_b \mathop{\mathbb{E}}_{k \leftarrow \lK} \Tr[\ketbra{b} A_{1, \lambda} \circ (\Phi^\enc _{k,b} \otimes \mathds{1}_S ) \circ A_{0, \lambda}(1)] \leq \frac{1}{2} + \eta(\lambda)
    \end{equation}
    where $\lambda$ is implicit on the left-hand side, $b \in \{0,1 \}$, and $\lK_{\lambda}$ is the random variable distributed on $\{0,1\}^{\kappa(\lambda)}$ such that
    \begin{equation}
        \Pr[\lK_\lambda = k] = \Tr[\ketbra{k} \Phi^\key _\lambda (1)].
    \end{equation}
\end{definition}

%% file: tex/3_security_definitions.tex
\section{Security Definitions}\label{sec:definitions}

In this section, we introduce a new description of the certified deletion security notion.
First, however, we must augment our QECM framework to allow it to detect errors on decryption.

\begin{definition}[Augmented Quantum Encryption of Classical Messages]
\label{def:aqecm}
    Let $n$ be an integer.
    Let $\lS = (\key, \enc, \dec)$ be an $n$-QECM.
    An $n$-\emph{augmented quantum encryption of classical messages} (n-AQECM) is a tuple of uniform efficient quantum circuits $\hat{\lS} = (\key, \enc, \widehat{\dec}$, where $\widehat{\dec}$ implements a CPTP map of the form
    \begin{equation}
        \Phi_{\lambda} ^{\widehat \dec} \colon \lD(\lH_{K, \lambda} \otimes \lH_{T, \lambda}) \rightarrow \lD( \lH_M \otimes \lQ ).
    \end{equation}
    For all $\lambda \in \bN^+, k \in \{0,1\}^{\kappa(\lambda)}$, and $m \in \{0,1\}^n$, the maps corresponding to the circuits must satisfy
    \begin{equation}
        \Tr[\ketbra{k} \Phi^\key(1)] > 0 \Rightarrow \Tr[ \ketbra{m} \otimes \ketbra{1} \Phi^{\widehat \dec} _k \circ \Phi^\enc _k \ketbra{m}] = 1,
    \end{equation}
    where $\lambda$ is implicit, $\Phi_k ^\enc$ is the CPTP map defined by $\rho \mapsto \Phi^\enc (\ketbra{k} \otimes \rho )$, and we define $\Phi^\dec _k$ analogously.
\end{definition}
The extra qubit (which will be referred to as a flag), though by itself without any apparent use, may serve as a way to indicate that the decryption process did not proceed as expected in any given run.
In the case of decryption without error, the circuit should output $\ketbra{1}$, and in the case of decryption error, the circuit should output $\ketbra{0}$.
This allows us to define a criterion by which an AQECM might be robust against a certain amount of noise.

Since the original QECM framework will no longer be used for the rest of this paper, we henceforth note that all further references to the QECM framework are in fact references to the AQECM framework.

\begin{definition}[Robust Quantum Encryption of Classical Messages]
    \label{def:robust}
    Let $\lS = (\key, \enc, \dec)$ be an $n$-QECM.
    We say that $\lS$ is $\epsilon$-\emph{robust} if, for all adversaries $\lA$ implementing CPTP maps of the form
    \begin{equation}
        A \colon \lD( \lH_{T, \lambda} ) \rightarrow \lD( \lH_{T, \lambda} ),
    \end{equation}
    and for two distinct messages $m, m' \in \lH_M$, we have that
    \begin{equation}
        \mathop{\mathbb{E}}_{k \leftarrow \lK} \Tr[\ketbra{m'} \otimes \ketbra{1} \Phi_k ^\dec \circ A \circ \Phi_k ^\enc \ketbra{m}] \leq \epsilon.
    \end{equation}
\end{definition}
In other words, a QECM is $\epsilon$-robust if, under interference by an adversary, the event that decryption yields a different message than was encrypted and that the decryption circuit approves of the outcome is less than or equal to $\epsilon$.
This is functionally equivalent to a one-time quantum authentication scheme, where messages are classical (see \emph{e.g.}~\cite{BCG+02,GYZ17,DNS12}).

Our description takes the form of an augmentation of the QECM framework described in \cref{def:aqecm}.
Given a QECM with key $k$ and encrypting message $m$, the certified deletion property should guarantee that the recipient, Bob, cannot do the following two things simultaneously:
\begin{itemize}
    \item Make Alice, the sender, accept his certificate of deletion; and
    \item Given $k$, recover information about $m$.
\end{itemize}
\begin{definition}[Certified Deletion Encryption]\label{def:cde}
    Let $\lS = (\key, \enc, \dec)$ be an $n$-QECM such that
    Let $\del$ and $\ver$ be efficient quantum circuits implemented by CPTP maps of the form
    \begin{itemize}
        \item $\Phi^\del _{\lambda} \colon \lD(\lH_{T, \lambda}) \rightarrow \lD(\lH_{D, \lambda})$
        \item $\Phi^\ver _{\lambda} \colon \lD(\lH_{K, \lambda} \otimes \lH_{D, \lambda}) \rightarrow \lD(\lQ)$
    \end{itemize}
    where $\lH_{D, \lambda} = \lQ(d(\lambda))$ for a function $d \colon \bN^+ \rightarrow \bN^+$.

    For all $\lambda \in \bN^+$, $k \in \{0,1\}^{\kappa(\lambda)}$, and $m \in \{0,1\}^n$, the maps must satisfy
    \begin{equation} \label{eq:cde2}
        \Tr[\ketbra{k} \Phi^{\key}(1)] > 0 \implies \Tr[\ketbra{1} \Phi^\ver \circ \left(\ketbra{k} \otimes \left(\Phi^\del \circ \Phi^\enc _k \ketbra{m} \right) \right)] = 1
    \end{equation}
    where $\lambda$ is implicit.

    We call the tuple $\lS' = (\key, \enc, \dec, \del, \ver)$ an $n$-\emph{certified deletion encryption} ($n$-CDE).
\end{definition}

\begin{definition}[Certified Deletion Attack]
    Let $\lS = (\key, \enc, \dec, \del, \ver)$ be an $n$-CDE.
    A \emph{certified deletion attack} is a quantum adversary $\lA = (\lA_0, \lA_1, \lA_2)$ implementing CPTP maps of the form
    \begin{itemize}
        \item $A_{0, \lambda} \colon \lD(\bC) \rightarrow \lD(\lH_{M} \otimes \lH_{S, \lambda})$,
        \item $A_{1, \lambda} \colon \lD(\lH_{T, \lambda} \otimes \lH_{S, \lambda}) \rightarrow \lD(\lH_{D, \lambda} \otimes \lH_{S, \lambda} \otimes \lH_{T', \lambda})$, and
        \item $A_{2, \lambda} \colon \lD(\lH_{K, \lambda} \otimes \lH_{S, \lambda} \otimes \lH_{T', \lambda}) \rightarrow \lD(\lQ)$
    \end{itemize}
    where $\lH_{S, \lambda} = \lQ(s(\lambda))$ and $\lH_{T', \lambda} = \lQ(\ell'(\lambda))$ for functions $s, \ell' \colon \bN^+ \rightarrow \bN^+$.
\end{definition}

We are now ready to define our notion of certified deletion security. We refer the reader to \cref{sec:intro-definitions} for an informal explanation of the definition, and we recall that notation
$\Phi^\enc_{k,b}$ is defined in \cref{eq:qecm}.

\begin{definition}[Certified Deletion Security]\label{def:cds}
    Let $\lS = (\key, \enc, \dec, \del, \ver)$ be an $n$-CDE.
    For any fixed and implicit $\lambda \in \bN^+$, we define the CPTP map $\Phi^\ver_{k} \colon \lD(\lH_{K, \lambda}\ otimes \lH_{D, \lambda}) \rightarrow \lD(\lQ \otimes \lH_{K, \lambda})$ by
    \begin{equation}
        \rho \mapsto \Phi^\ver(\ketbra{k} \otimes \rho) \otimes \ketbra{k}.
    \end{equation}
    Let $b \in \{0,1\}$, let $\lA$ be a certified deletion attack, and let
    \begin{equation}
        \label{eq:cdsec}
        p_b = \mathop[{\mathbb{E}}_{k \leftarrow \lK} \Tr[(\ketbra{1, 1}) (\mathds{1} \otimes A_{2}) \circ (\Phi^\ver _k \otimes \mathds{1}_{ST'}) \circ A_{1} \circ (\Phi^\enc_{k,b} \otimes \mathds{1}_S ) \circ A_{0} (1)],
    \end{equation}
    where $\lambda$ is implicit, and where $\lK_{\lambda}$ is the random variable distributed on $\{0,1\}^{\kappa(\lambda)}$ such that
    \begin{equation}
        \Pr[\lK_\lambda = k] = \Tr[\ketbra{k} \Phi^\key_\lambda (1)].
    \end{equation}
    Then we say that $\lS$ is $\eta$-\emph{certified deletion secure} if, for all certified deletion attacks~$\lA$, there exists a negligible function~$\eta$ such that
    \begin{align}
        \label{eq:cdsec}
        \abs{p_0 - p_1} \leq \eta(\lambda).
    \end{align}
\end{definition}

%% file: tex/4_construction.tex
\section{Constructing an Encryption Scheme with Certified Deletion}
\label{sec:pm}

\cref{sch:pmcd} aims to exhibit a noise-tolerant prepare-and-measure $n$-CDE with ciphertext indistinguishability and certified deletion security.

\begin{table}
    {\small
    \begin{tabular}{ll}
        \nc{$M_A ^{\theta, x}$}{Measurement operator acting on system $A$ with setting $\theta$ and outcome $x$}
        \nc{$\lM_{A \rightarrow X|S^\Theta} ^\lI$}{Measurement map applied on the qubits of system $A$ indexed by $\lI$, with setting $S^\Theta$,}
        \nc{}{and outcome stored in register $X$}
        \hline
        \nc{$\lambda$}{Security parameter}
        \nc{$n$}{Length, in bits, of the message}
        \nc{$m = \kappa(\lambda)$}{Total number of qubits sent from encrypting party to decrypting party}
        \nc{$k$}{Length, in bits, of the string used for verification of deletion}
        \nc{$s = m - k$}{Length, in bits, of the string used for extracting randomness}
        \nc{$\tau = \tau(\lambda)$}{Length, in bits, of error correction hash}
        \nc{$\mu = \mu(\lambda)$}{Length, in bits, of error syndrome}
        \nc{$\theta$}{Basis in which the encrypting party prepares her quantum state}
        \nc{$\delta$}{Threshold error rate for the verification test}
        \hline
        \nc{$\Theta$}{Set of possible bases from which $\theta$ is chosen}
        \nc{$\fH_{\pa}$}{Universal$_2$ family of hash functions used in the privacy amplification scheme}
        \nc{$\fH_{\ec}$}{Universal$_2$ family of hash functions used in the error correction scheme}
        \nc{$H_{\pa}$}{Hash function used in the privacy amplification scheme}
        \nc{$H_{\ec}$}{Hash function used in the error correction scheme}
        \nc{$S^\Theta$}{Seed for the choice of $\theta$}
        \nc{$S^{H_\pa}$}{Seed for the choice of the hash function used in the error correction scheme}
        \nc{$S^{H_\ec}$}{Seed for the choice of the hash function used in the privacy amplification scheme}
        \nc{$\synd$}{Function that computes the error syndrome}
        \nc{$\corr$}{Function that computes the corrected string}
    \end{tabular}
    }
    \caption{Overview of nomenclature used in~\cref{sec:pm} and \cref{sec:sec}}
    \label{tab:pm-nom}
\end{table}

\begin{scheme}[Prepare-and-Measure Certified Deletion] \label{sch:pmcd}
    Let $n, \lambda, \tau, \mu, m = s + k$ be integers.
    Let $\Theta = \{ \theta \in \{0,1\}^m \mid \omega(\theta) = k \}$.
    Let both $\fH_{\ec} \coloneqq \{ h \colon \{ 0,1 \}^s \rightarrow \{ 0,1 \}^\tau \}$ and $\fH_{\pa} \coloneqq \{ h \colon \{ 0, 1 \}^s \rightarrow \{0, 1\}^n \}$ be universal$_2$ families of hash functions.
    Let $\synd \colon \{0,1\}^n \rightarrow \{0,1\}^\mu$ be an error syndrome function, let $\corr \colon \{0,1\}^n \times \{0,1\}^\mu \rightarrow \{0,1\}^n$ be the corresponding function used to calculate the corrected string, and let $\delta \in [0,1]$ be a tolerated error rate for verification.
    We define a \emph{noise-tolerant prepare-and-measure $n$-CDE} by Circuits 1-5. This scheme satisfies Equation $\eqref{eq:cde2}$.
It is therefore an $n$-CDE.
\end{scheme}

    \begin{algorithm}
    \DontPrintSemicolon
    \caption{The key generation circuit $\key$.}
        \SetKwInOut{Input}{Input}
        \SetKwInOut{Output}{Output}

        \Input{None.}
        \Output{A key state $\rho \in  \lD (\lQ(k + m + n + \mu + \tau) \otimes \fH_\pa \otimes \fH_\ec))$.}

        Sample $\theta \unif \Theta$.\;
        Sample $r|_{\bar{\lI}} \unif \{0,1\}^k$ where $\bar{\lI} = \{ i \in [m] \mid \theta_i = 1 \}$.\;
        Sample $u \unif \{0,1\}^n$.\;
        Sample $d \unif \{0,1\}^\mu$.\;
        Sample $e \unif \{0,1\}^\tau$.\;
        Sample $H_{\pa} \unif \fH_{\pa}$.\;
        Sample $H_{\ec} \unif \fH_{\ec}$.\;
        Output $\rho = \ketbra{r|_{\bar{\lI}}, \theta, u, d, e, H_\pa, H_\ec}$.

    \end{algorithm}

    \begin{algorithm}
    \DontPrintSemicolon
    \caption{The encryption circuit $\enc$.}
        \SetKwInOut{Input}{Input}
        \SetKwInOut{Output}{Output}

        \Input{A plaintext state $\ketbra{\msg} \in \lD(\lQ(n))$ and a key state $\ketbra{r|_{\bar{\lI}}, \theta, u,
 d, e, H_{\pa}, H_{\ec}} \in \lD(\lQ(k + m + n + \mu + \tau) \otimes \fH_\pa \otimes \fH_\ec)$.}
        \Output{A ciphertext state $\rho \in \lD(\lQ(m+n+\tau+\mu))$.}

        Sample $r|_{\lI} \unif \{0,1\}^s$ where $\lI = \{ i \in [m] \mid \theta_i = 0 \}$.\;
        Compute $x = H_{\pa} (r|_\lI)$ where $\lI = \{i \in [m] \mid \theta_i = 0\}$.\;
        Compute $p = H_{\ec} (r|_\lI) \oplus d$.\;
        Compute $q = \synd(r|_\lI) \oplus e$.\;
        Output $\rho = \ketbra{r^\theta} \otimes \ketbra{\msg \oplus x \oplus u, p, q}$.\;

    \end{algorithm}

    \begin{algorithm}
    \DontPrintSemicolon
    \caption{The decryption circuit $\dec$.}
        \SetKwInOut{Input}{Input}
        \SetKwInOut{Output}{Output}

        \Input{A key state $\ketbra{r|_{\bar{\lI}}, \theta, u, d, e, H_{\pa}, H_{\ec}} \in \lD(\lQ(k + m + n + \mu +
         \tau) \otimes \fH_\pa \otimes \fH_\ec)$ and a ciphertext $\rho \otimes \ketbra{c, p, q} \in \lD(\lQ(m + n + \mu + \tau))$.}
        \Output{A plaintext state $\sigma \in \lD(\lQ(n))$ and an error flag $\gamma \in \lD(\lQ)$.}

        Compute $\rho' = \sH^\theta \rho \sH^\theta$.\;
        Measure $\rho '$ in the computational basis.
Call the result $r$.\;
        Compute $r' = \corr(r|_\lI, q \oplus e)$ where $\lI = \{i \in [m] \mid \theta_i = 0\}$.\;
        Compute $p' = H_\ec(r') \oplus d$.\;
        If $p \neq p'$, then set $\gamma = \ketbra{0}$.
        Else, set $\gamma = \ketbra{1}$.\;
        Compute $x' = H_\pa(r')$.\;
        Output $\sigma \otimes \gamma = \ketbra{c \oplus x' \oplus u} \otimes \gamma$.\;
    \end{algorithm}

    \begin{algorithm}
    \DontPrintSemicolon
    \caption{The deletion circuit $\del$.}
        \SetKwInOut{Input}{Input}
        \SetKwInOut{Output}{Output}

        \Input{A ciphertext $\rho \otimes \ketbra{c, p, q} \in \lD(\lQ(m + n + \mu + \tau))$.}
        \Output{A certificate string state $\sigma \in \lD(\lQ(m))$.}

        Measure $\rho$ in the Hadamard basis.
        Call the output $y$.\;
        Output $\sigma = \ketbra{y}$.
    \end{algorithm}

    \begin{algorithm}
    \DontPrintSemicolon
    \caption{The verification circuit $\ver$.}
        \SetKwInOut{Input}{Input}
        \SetKwInOut{Output}{Output}

        \Input{A key state $\ketbra{r|_{\bar{\lI}}, \theta, u, d, e, H_{\pa}, H_{\ec}} \in \lD(\lQ(k + m + n + \mu +
         \tau) \otimes \fH_{\pa} \otimes \fH_{\ec})$ and a certificate string state $\ketbra{y} \in \lD(\lQ(m))$.}
        \Output{A bit.}

        Compute $\hat y' = \hat y|_{\bar{\lI}}$ where $\bar{\lI} = \{i \in [m] \mid \theta_i = 1\}$.\;
        Compute $q = r|_{\bar{\lI}}$.\;
        If $\omega(q \oplus \hat y') < k \delta,$ output $1$.
        Else, output $0$.\;
    \end{algorithm}

%% file: tex/5_security_analysis.tex
\section{Security Analysis}
\label{sec:sec}

In this section, we present the security analysis for \cref{sch:pmcd}: in \cref{sec:indi-security}, we show the security of the scheme in terms of an encryption scheme, then, in~\cref{sec:correctness}, we show that the scheme is correct and robust. Finally in \cref{ssec:cdsec}, we show that the scheme is a certified deletion scheme.

\subsection{Ciphertext Indistinguishability}\label{sec:indi-security}

In considering whether~\cref{sch:pmcd} has ciphertext indistinguishability (\cref{def:indistinguishable}), one need only verify that an adversary, given a ciphertext, would not be able to discern whether a known message was encrypted.
\begin{theorem}
    \cref{sch:pmcd} has ciphertext indistinguishability.
\end{theorem}
\begin{proof}
    For any distinguishing attack $\lA = (\lA_0, \lA_1)$, any state $\rho = \rho_S \otimes \ketbra{\msg} \in \lD(\lH_S \otimes \lQ(n))$, and where $k = (r, \theta, u, d, e, H_\pa, H_\ec) \in \{ 0,1 \}^{m + n + \mu + \tau} \times \fH_\pa \times \fH_\ec$ is a key, we have that
    \begin{align*}
        \E_k \left(\dOne_S \otimes \Phi ^\enc _{k,1} \right) (\rho) &= \frac{1}{2^{m + n + \mu + \tau} |\fH_\pa| |\fH_\ec|} \sum_{k} \rho_S \otimes \ketbra{r^\theta} \otimes \ketbra{\msg \oplus x \oplus u, p, q} \\
        &= \frac{1}{2^{m + n + \mu + \tau} |\fH_\pa| |\fH_\ec|} \sum_{k} \rho_S \otimes \ketbra{r^\theta} \otimes \ketbra{x \oplus u, p, q} \\
        &= \E_k \left(\dOne_S \otimes \Phi ^\enc _{k,0} \right) (\rho),
    \end{align*}
    where the second equality is due to the uniform distribution of both $\msg \oplus x \oplus u$ and $u$.
    Therefore, an adversary can do no better than guess $b$ correctly half of the time in a distinguishing attack.
    This implies perfect ciphertext indistinguishability with $\eta = 0$.
\end{proof}

\subsection{Correctness}
\label{sec:correctness}

Thanks to the syndrome and correction functions included in the scheme, the decryption circuit is robust against a certain amount of noise; that is, below such a level of noise, the decryption circuit outputs Alice's original message with high probability.
This noise threshold is determined by the distance of the linear code used.
In particular, where $\Delta$ is the distance of the code, decryption should proceed normally as long as fewer than $\lfloor \frac{\Delta - 1}{2} \rfloor$ errors occur to the quantum encoding of $r|_\lI$ during transmission through the quantum channel.

To account for greater levels of noise (such as may occur in the presence of an adversary), we show that the error correction measures implemented in~\cref{sch:pmcd} ensure that errors in decryption are detected with high probability.
In other words, we show that the scheme is $\epsilon_\rob$-robust, where $\epsilon_\rob \coloneqq \frac{1}{2^\tau}$.

Recall that $\tau$ is the length of the error correction hash, and that $\mu$ is the length of the error correction syndrome. 
Consider that Bob has received a ciphertext state $\rho_B \otimes \ketbra{c,p,q} \in \lD(\lQ(m + n + \mu + \tau))$ and a key $(r|_{\bar{\lI}}, \theta, u, d, e, H_\pa, H_\ec) \in \Theta \times \{0,1\}^{n + \mu + \tau} \times \fH_\pa \times \fH_\ec$.
Given $\theta$, Bob learns $\lI$.
This allows him to perform the following measurement on $\rho_B$:
\begin{equation}
    \lM_{B \rightarrow Y} ^\lI (\cdot) = \sum_{y \in \{0,1\}^s} \ket{y}_Y  \left( M_{B_\lI} ^{0, y} \right) \cdot \left( M_{B_\lI} ^{0, y} \right)^\dagger \bra{y}_Y
\end{equation}
The new register $Y$ contains a hypothesis of the random string Alice used in generating $c$.
Since $\rho_B$ was necessarily transmitted through a quantum channel, it may have been altered due to noise.
Bob calculates a corrected estimate: $\hat{x} = \corr(y, q \oplus e)$.
Finally, he compares a hash of the estimate with $p \oplus d$, which is the hash of Alice's corresponding randomness.
This procedure is represented by a function $\ec \colon \{0,1\}^s \times \{0,1\}^\mu \times \fH_\ec \rightarrow \{0,1\}$ defined by
\begin{equation}
    \ec(x,y) =
    \begin{cases}
    0 \quad &\textrm{if}~H_\ec (x) \neq y\\
    1 \quad &\textrm{else.}
    \end{cases}
\end{equation}
To record the value of this test, we use a flag $F^\ec \coloneqq \ec(\hat{x}, p \oplus d)$.
It is very unlikely that both $F^\ec = 1$ and the outcome of Bob's decryption procedure is not equal to Alice's originally intended message.
This is shown in the following proposition, the proof of which follows that of an analogous theorem in~\cite{TL17}.

\begin{theorem}
    If $r|_{\lI} \in \{0,1\}^m$ is the random string Alice samples in encryption, and $\hat{x} = \corr(y, q \oplus e)$, then
    \begin{equation}
        \Pr[H_\pa(r|_\lI) \neq H_\pa(\hat{x}) \wedge F^\ec = 1] \leq \frac{1}{2^\tau}.
    \end{equation}
\end{theorem}
\begin{proof}
\begin{align}
    \Pr[H_\pa(r|_\lI) \neq H_\pa(\hat{x}) \wedge F^\ec = 1]
    &= \Pr[H_\pa(r|_\lI) \neq H_\pa(\hat{x}) \wedge H_\ec(p \oplus d) = H_\ec(\hat{x})]\\
    &= \Pr[H_\pa(r|_\lI) \neq H_\pa(\hat{x}) \wedge H_\ec(r|_\lI) = H_\ec(\hat{x})]\\
    &\leq \Pr[r|_\lI \neq \hat{x} \wedge H_\ec(r|_\lI) = H_\ec(\hat{x})]\\
    &= \Pr[r|_\lI \neq \hat{x}] \Pr[H_\ec(r|_\lI) = H_\ec(\hat{x})]\\
    &\leq \Pr[H_\ec(r|_\lI) = H_\ec(\hat{x}) \mid r|_\lI \neq \hat{x}]\\
    &\leq \frac{1}{ \lVert \fH_\ec \rVert }\\
    &= \frac{1}{2^\tau}.\qedhere
\end{align}
\end{proof}

\subsection{Certified Deletion Security}\label{ssec:cdsec}

We now prove certified deletion security of~\cref{sch:pmcd}.
Our technique consists in formalizing a game (\cref{game:pm}) that corresponds to the security definition (\cref{def:cds}) applied to \cref{sch:pmcd}.
Next, we develop an entanglement-based sequence of interactions (\cref{game:epr}) which accomplish the same task as in the previous Game.
We analyze this game and afterwards we show formally that the aforementioned analysis, via its relation to~\cref{game:pm}, implies the certified deletion security of~\cref{sch:pmcd}.
To begin, we describe a game which exhibits a certified deletion attack on~\cref{sch:pmcd}, and which thus allows us to examine whether the scheme has certified deletion security.
In what follows, the challenger represents the party who would normally encrypt and send the message (Alice), and the adversary $\lA$ represents the recipient (Bob).
The adversary sends the challenger a candidate message $\msg_0 \in \{0,1\}^n$ and Alice chooses, with uniform randomness, whether to encrypt $0^n$ or $\msg_0$; security holds if, for any adversary, the probabilities of the following two events are negligibly close:
\begin{itemize}
    \item verification passes \emph{and} Bob outputs 1, in the case that Alice encrypted $0^n$;
    \item verification passes \emph{and} Bob output 1,  in the case that Alice encrypted $\msg_0$.
\end{itemize}

\begin{game}[Prepare-and-Measure Game]\label{game:pm}
    Let $\lS = (\key, \enc, \dec, \del, \ver)$ be an $n$-CDE with $\lambda$ implicit, and with circuits defined as in~\cref{sch:pmcd}.
    Let $\lA = (\lA_0, \lA_1, \lA_2)$ be a certified deletion attack.
    The game is parametric in $b \unif \{0,1\}$ and is called \cref{game:pm}$(b)$..
    \begin{enumerate}
        \item Run $\ketbra{\msg_0}_M \otimes \rho_S \leftarrow A_0 (1)$.
        Generate
        \begin{equation}
            \ketbra{\theta, u, d, e, H_{\pa}, H_{\ec}, r|_{\bar{\lI}}}_K \leftarrow \Phi^\key.
        \end{equation}
        Denote
        \begin{equation}
            \msg \coloneqq \begin{cases}
            0^n \quad &\textrm{if}~b = 0\\
            \msg_0 \quad &\textrm{if}~b = 1.
            \end{cases}
        \end{equation}
        Compute
        \begin{align}
            \begin{split}
                \ketbra{r^\theta}_T &\otimes \ketbra{\msg \oplus x \oplus u, p, q}_T \\
                &\leftarrow \Phi^\enc (\ketbra{\theta, u, d, e, H_{\pa}, H_{\ec}, r|_{\bar{\lI}}}_K \otimes \ketbra{\msg}_M).
            \end{split}
        \end{align}
        \item Run
            \begin{align}
                \ketbra{y}_D &\otimes \rho'_S \otimes \rho_{T'} \leftarrow A_1 (\ketbra{r^\theta}_T \otimes \ketbra{\msg \oplus x \oplus u, p, q}_T \otimes \rho_S).
            \end{align}
        Compute
        \begin{equation}
            \Gamma(ok) \leftarrow \Phi^\ver (\ketbra{\theta, u, d, e, H_{\pa}, H_{\ec}, r|_{\bar{\lI}}}_K \otimes \ketbra{y}_D).
        \end{equation}
        \item If $ok = 1$, run
            \begin{equation}
                \ketbra{b'} \leftarrow A_2 (\ketbra{r|_{\bar{\lI}}, \theta, u, d, e, H_{\pa}, H_{\ec}}_{K'} \otimes \rho'_S \otimes \rho_{T'});
            \end{equation}
            else, $b' \coloneqq 0$.
    \end{enumerate}
\end{game}

Let $p_b$ be the probability that the output of \cref{game:pm}$(b)$ is~1. Comparing \cref{game:pm} with \cref{def:cds}, we note that the former runs the adversary to the end only in the case that $ok = 1$, while the latter runs the adversary to the end in both cases. However, the obtained distribution for  $p_b$ is the same, since in \cref{game:pm}, $p_b=1$ whenever the adversary outputs~$1$ \emph{and} $ok = 1$. Hence we wish to bound $\abs{p_0 - p_1}$ in \cref{game:pm}.
Instead of directly analyzing~\cref{game:pm}, we analyze a game wherein the parties use entanglement; this allows us to express the game in a format that is conducive for the analysis that follows.

\begin{game}[EPR Game]\label{game:epr}
    Alice is the sender, and Bob is the recipient and adversary.
    The game is parametric in $b \unif \{0,1\}$ and is called \cref{game:epr}$(b)$.
    \begin{enumerate}
        \item  \label{epr-game-step1}
        Bob selects a string~$\msg_0 \in \{0,1\}^n$ and sends~$\msg_0$ to Alice.
        Bob prepares a tripartite state~$\rho_{ABB'} \in \lD(\lQ(3m))$ where each system contains~$m$ qubits.
        Bob sends the $A$ system to Alice and keeps the systems $B$ and $B'$.
        Bob measures the $B$ system in the Hadamard basis and obtains a string~$y \in \{0,1\}^m$.
        Bob sends $y$ to Alice.
        \item \label{epr-game-step2}
            Alice samples $\theta \unif \Theta$, $r|_{\bar{\lI}} \unif \{0,1\}^k$, $u \unif \{0,1\}^n$, $d \unif \{0,1\}^\mu$, $e \unif \{0,1\}^\tau$, $H_{\pa} \unif \fH_\pa$, and $H_\ec \unif \fH_\ec$.
        She applies a CPTP map to system $A$ which measures~$A_i$ according to the computational basis if $\theta_i = 0$ and the Hadamard basis if $\theta_i = 1$.
        Call the result $r$.
        Let $\lI = \{i \in [m] \mid \theta_i = 0\}$.
        Alice computes $x = H_\pa (r|_\lI), p = H_\ec (r|_\lI) \oplus d,$ and $q = \synd(r|_\lI) \oplus e$.
        Alice selects a message:
        \begin{equation}
            \msg \coloneqq \begin{cases}
            0^n \quad &\textrm{if}~b = 0\\
            \msg_0 \quad &\textrm{if}~b = 1.
            \end{cases}
        \end{equation}
        If $\omega( y \oplus r|_{\bar{\lI}} ) < k \delta$, $ok \coloneqq 1$ and Alice sends \begin{equation}(\msg \oplus x \oplus u, r|_{\bar{\lI}}, \theta, u, d, e, p, q, H_\pa, H_\ec)\end{equation} to Bob.
        Else, $ok \coloneqq 0$ and $b \coloneqq 0$.
        \item \label{epr-game-step3}
        If $ok = 1$, Bob computes
        \begin{align}
            \begin{split}
                &\ketbra{b'} \leftarrow \lE (\rho_{B'} \otimes \\&\ketbra{\msg \oplus x \oplus u, \msg_0, r_{\bar{\lI}}, \theta, u, d, e, p, q, H_\pa, H_\ec})
            \end{split}
        \end{align} for some CPTP map $\lE$; else $b' \coloneqq 0$.
    \end{enumerate}
\end{game}

\cref{game:epr} is intended to model a purified version of~\cref{game:pm}.
Note that Bob's measuremement of~$B$ in the Hadamard basis is meant to mimic the~$\del$ circuit of~\cref{sch:pmcd}.
Although it may seem strange that we impose a limitation of measurement basis on Bob here, it is in fact no limitation at all; indeed, since Bob prepares~$\rho_{ABB'}$, he is in total control of the state that gets measured, and hence may assume an arbitrary degree of control over the measurement outcome.
Therefore, the assumption that he measures in the Hadamard basis is made without loss of generality.

It may also appear that the adversary in~\cref{game:pm} has more information when producing the deletion string than Bob in~\cref{game:epr}. This, however, is not true, as the adversary in~\cref{game:pm} has only received information from Alice that appears to him to be uniformly random (as mentioned, the statement is formalized later, in \cref{ss:secred}).
In order to further the analysis, we assign more precise notation for the maps described in~\cref{game:epr}.

\paragraph*{Bob's measurements.}

Measurement of Bob's system~$B$ of~$m$ qubits in Step~\ref{epr-game-step1} is represented using two CPTP maps: one acting on the systems in~$\lI$, with outcome recorded in register~$Y$; and one acting on the systems in~$\bar \lI$, with outcome recorded in~$W$.
Note, however, that Bob has no access to~$\theta$, and therefore has no way of determining~$\lI$.
The formal separation of registers~$Y$ and~$W$ is simply for future ease of specifying the qubits to which we refer.

Recall the definition of the measurements $M_B ^{x,y}$ from~\cref{ss:prelim-msmt}.

The first measurement, where the outcome is stored in register~$Y$, is defined~by
\begin{equation}
    \lM_{B \rightarrow Y} ^\lI (\cdot) = \sum_{y \in \{0,1\}^s} \ket{y}_Y  \left( M_{B_\lI} ^{1, y} \right) \cdot \left( M_{B_\lI} ^{1, y} \right)^\dagger \bra{y}_Y
\end{equation}
and the second, where the outcome is stored in register~$W$, is defined by
\begin{equation}
    \lM_{B \rightarrow W} ^{\bar \lI} ( \cdot ) = \sum_{w \in \{0,1\}^k} \ket{w}_W \left ( M_{B_{\bar \lI}} ^{1, w} \right) \cdot \left( M_{B_{\bar \lI}} ^{1, w} \right)^\dagger \bra{w}_W,
\end{equation}
where~$M_{B_\lI} ^{1, y} \coloneqq \bigotimes_{i \in \lI} M_{B_i} ^{1, y_i}$, and the definition of~$M_{B_{\bar \lI}} ^{1, w}$ is analogous.

\paragraph*{Alice's measurements.}

We represent the randomness of Alice's sampling using seed registers.
Thus, the randomness used for Alice's choice of basis is represented as
\begin{equation}
    \rho_{S^\Theta} = \frac{1}{\binom mk} \sum_{\theta \in \Theta} \ketbra{\theta}_{S^\Theta}.
\end{equation}
Similarly, Alice's randomness for choice of a hash function for privacy amplification is represented as
\begin{equation}
    \rho_{S^{H_\pa}} = \frac{1}{|\fH_\pa|} \sum_{h \in \fH_\pa} \ketbra{h}_{S^{H_\pa}}.
\end{equation}

Recall that $m = s + k$, where $k$ is the weight of all strings in $\Theta$.
Measurement of Alice's system $A$ of $m$ qubits in Step~\ref{epr-game-step2} is represented using two CPTP maps: one acting on the systems in $\lI$, with outcome recorded in register~$X$ (by definition, these qubits are measured in the computational basis); and one acting on the systems in $\bar \lI$, with outcome recorded in register~$V$ (by definition, these qubits are measured in the Hadamard basis).
\begin{equation*}
    \lM_{A \rightarrow X | S^\Theta} ^\lI (\cdot) = \sum_{\theta \in \Theta} \sum_{x \in \{0,1\}^s} \ket{x}_X  \left( M_{A_\lI} ^{0, x} \otimes \ketbra{\theta}_{S^\Theta} \right) \cdot \left( M_{A_\lI} ^{0, x} \otimes \ketbra{\theta}_{S^\Theta} \right)^\dagger \bra{x}_X;
\end{equation*}
and the second measurement, where the outcome is stored in register $V$, is defined~by
\begin{equation*}
    \lM_{A \rightarrow V | S^\Theta} ^{\bar \lI} ( \cdot ) = \sum_{\theta \in \Theta} \sum_{v \in \{0,1\}^k} \ket{v}_V \left ( M_{A_{\bar \lI}} ^{1, v} \otimes \ketbra{\theta}_{S^\Theta} \right) \cdot \left( M_{A_{\bar \lI}} ^{1, v} \otimes \ketbra{\theta}_{S^\Theta} \right)^\dagger \bra{v}_V,
\end{equation*}
where $M_{A_\lI} ^{0, x} \coloneqq \bigotimes_{i \in \lI} M_{A_i} ^{0, x_i}$ and the definition of $M_{A_{\bar \lI}} ^{1, v}$ is analogous.

We also introduce a hypothetical measurement for the sake of the security analysis.
Consider the case where Alice measures all of her qubits in the Hadamard basis.
In this case, instead of $\lM_{A \rightarrow X | S^\Theta} ^\lI$, Alice would use the measurement
\begin{equation*}
    \lM_{A \rightarrow Z | S^\Theta} ^\lI (\cdot) = \sum_{\theta \in \Theta} \sum_{z \in \{0,1\}^s} \ket{z}_Z  \left( M_{A_\lI} ^{1, z} \otimes \ketbra{\theta}_{S^\Theta} \right) \cdot \left( M_{A_\lI} ^{1, z} \otimes \ketbra{\theta}_{S^\Theta} \right)^\dagger \bra{z}_Z.
\end{equation*}

Each of Alice's and Bob's measurements commute with each other as they all act on distinct quantum systems.
We can thus define the total measurement map
\begin{equation}
    \lM_{AB \rightarrow VWXY|S^\Theta} = \lM_{A \rightarrow X|S^\Theta} ^\lI \circ \lM_{A \rightarrow V|S^\Theta} ^{\bar \lI} \circ \lM_{B \rightarrow Y} ^\lI \circ \lM_{B \rightarrow W} ^{\bar\lI}.
\end{equation}
The overall post-measurement state is denoted $\sigma_{VWXYS^\Theta}$. We analogously define the hypothetical post-measurement state $\hat{\sigma}_{VWZYS^\Theta}$.

\paragraph*{Alice's verification:}

Alice completes the verification procedure by comparing the~$V$ register to the $W$ register.
If they differ in less than $k \delta$ bits, then the test is passed.
The test is represented by a function $\comp \colon \{0,1\}^k \times \{0,1\}^k \rightarrow \{0,1\}$ defined by
\begin{equation}
    \comp(v,w) =
    \begin{cases}
    0 \quad &\textrm{if}~\omega(v \oplus w) \geq k \delta\\
    1 \quad &\textrm{else.}
    \end{cases}
\end{equation}
To record the value of this test, we use a flag $F^\comp \coloneqq \comp(v,w)$.

The import of the outcome of this comparison test is that if Bob is good at guessing Alice's information in the Hadamard basis, it is unlikely that he is good at guessing Alice's information in the computational basis.
This trade-off is represented in the uncertainty relation of~\cref{prop:unc}.

Note that we can define the post-comparison test state, since $A|_\lI$ is disjoint from $A|_{\bar \lI}$ and $B|_\lI$ is disjoint from $B|_{\bar \lI}$.
The state is denoted $\tau_{ABVWS^\Theta \mid F^\comp = 1}$.

The following proposition shows that in order to ensure that Bob's knowledge of $X$ is limited after a successful comparison test, and receiving the key, his knowledge about Alice's hypothetical Hadamard measurement outcome must be bounded below.

\begin{proposition}
\label{prop:appunc}
    Let $\epsilon \geq 0$.
    Then
    \begin{equation}
        H_{\min} ^\epsilon (X \wedge F^\comp = 1 | VWS^\Theta B')_\sigma + H_{\max} ^\epsilon (Z \wedge F^\comp = 1 | Y)_{\sigma} \geq s.
    \end{equation}
\end{proposition}
    \begin{proof}
    We apply~\cref{prop:unc} to the state $\tau_{ABVWS^\Theta \mid F^\comp = 1}$.
    To do this, we equate $C = VWS^\Theta B'$ and $E = S^\Theta B$.
    Using the measurement maps $\lM_{A \rightarrow X \mid S^\Theta}$ and $\lM_{A \rightarrow Z \mid S^\Theta}$ as the POVMs and using $\{ \braket{\theta} \}$ as the projective measurement, applying $\cref{prop:unc}$ yields
    \begin{equation}
        H_{\min} ^\epsilon (X \wedge F^\comp = 1 | VWS^\Theta B')_\sigma + H_{\max} ^\epsilon (Z \wedge F^\comp = 1 | S^\Theta B)_\tau \geq s.
    \end{equation}
    We then apply the measurement map $\lM_{B \rightarrow Y \mid S^\Theta}$ and discard $S^\Theta$.
    Finally, by~\cref{prop:data}, we note that
    \begin{equation}
        H_{\max} ^\epsilon (Z \wedge F^\comp = 1 \mid S^\Theta B )_\tau \leq H_{\max} ^\epsilon (Z \wedge F^\comp = 1 \mid Y)_{\hat{\sigma}},
    \end{equation}
    which concludes the proof.
\end{proof}

In the spirit of~\cite{TL17}, we provide an upper bound for the max-entropy quantity, thus establishing a lower bound for the min-entropy quantity.

\begin{proposition}
\label{prop:maxent}
    Letting $\nu \in (0,1)$, we define
    \begin{equation} \label{eq:epsilon}
        \epsilon(\nu) \coloneqq \exp \left( \frac{-s k^2 \nu^2}{m (k+1)} \right).
    \end{equation}
    Then, for any $\nu \in (0, \frac{1}{2} - \delta]$ such that $\epsilon(\nu)^2 < \Pr[F^\comp = 1]_\sigma = \Pr[F^\comp = 1]_{\hat{\sigma}}$,
    \begin{equation}
        H_{\max} ^{\epsilon(\nu)} (Z \wedge F^\comp = 1 \mid Y)_{\hat{\sigma}} \leq s \cdot h(\delta + \nu)
    \end{equation}
    where
    \begin{equation}
        h(x) \coloneqq -x \log x - (1-x) \log(1-x).
    \end{equation}
\end{proposition}
\begin{proof}
    Define the event
    \begin{equation}
        \Omega \coloneqq \begin{cases}
        1 \quad &\textrm{if}~\omega(Z \oplus Y) \geq s(\delta + \nu)\\
        0 \quad &\textrm{else.}
        \end{cases}
    \end{equation}
    Using~\cref{lem:stat1}, we get that
    \begin{align}
        \Pr\left[F^\comp = 1 \wedge \Omega\right]_{\hat{\sigma}} &= \Pr\left[\omega(V \oplus W) \leq k \delta \wedge \omega(Z \oplus Y) \geq s (\delta + \nu)\right]_{\sigma}\\ &\leq \epsilon(\nu)^2.
    \end{align}
    Given the state $\hat{\sigma}_{ZYF^\comp = 1}$, we use~\cref{lem:stat2} to remove the possibility of $\Omega$ and arrive at the smoothed state $\tilde{\sigma}_{ZYF^\comp}$ with $\Pr[\Omega]_{\tilde{\sigma}} = 0$ and
    \begin{equation}
        P(\hat{\sigma}_{ZYF^\comp=1}, \tilde{\sigma}_{ZYF^\comp}) \leq \epsilon(\nu).
    \end{equation}
    Since $\Pr[F^\comp = 1]_{\tilde{\sigma}} = 1$, we get that
    \begin{equation}
        H_{\max}^{\epsilon(\nu)} (Z \wedge F^\comp = 1 \mid Y)_{\hat{\sigma}} \leq H_{\max} (Z \wedge F^\comp = 1 \mid Y)_{\tilde{\sigma}} = H_{\max} (Z \mid Y)_{\tilde{\sigma}}.
    \end{equation}
    Expanding this conditional max-entropy~\cite[Sec.~4.3.2]{Tom12}, we obtain
    \begin{align}
        H_{\max} (Z \mid Y)_{\tilde{\sigma}} &= \log \left( \sum_{y \in \{0,1\}^s} \Pr[Y = y]_{\tilde{\sigma}} 2^{H_{\max} (Z \mid Y)_{\tilde{\sigma}}} \right)\\
        &\leq \max_{\substack{y \in \{0,1\}^s \\ \Pr[Y = y]_{\tilde{\sigma}} > 0}} H_{\max} (Z \mid Y = y)_{\tilde{\sigma}}\\
        &\leq \max_{\substack{y \in \{0,1\}^s \\ \Pr[Y = y]_{\tilde{\sigma}} > 0}} \log \left| \left\{ z \in \{0,1\}^s \colon \Pr[Z = z \mid Y = y]_{\tilde{\sigma}} > 0 \right\} \right|\\
        &= \max_{y \in \{0,1\}^s} \log \left| \left\{ z \in \{0,1\}^s \colon \Pr[Z = z \wedge Y = y]_{\tilde{\sigma}} > 0 \right\} \right|.
    \end{align}
    Since $\Pr[\Omega]_{\tilde{\sigma}} = 0$, we have
    \begin{align}
            \vert \{ z \in \{0,1\}^s \colon \Pr[Z = z \wedge Y = y]_{\tilde{\sigma}} > 0 \} \vert &\leq \vert \{ z \in \{0,1\}^s \colon \omega(z \oplus y) < s(\delta + \nu) \} \vert \\
        &= \sum_{\gamma = 0} ^{\lfloor s (\delta + \nu) \rfloor} \binom s \gamma.
    \end{align}
    When $\delta + \nu \leq 1/2$ (see~\cite[Sec.~1.4]{vLvdG12}), we have that $\sum_{\gamma = 0} ^{\lfloor s(\delta = \nu) \rfloor} \binom s \gamma \leq 2^{s \cdot h(\delta + \nu)}$.
\end{proof}

At this point, we use~\cref{prop:lhl}, the Leftover Hashing Lemma, to turn the min-entropy bound into a statement about how close to uniformly random the string $\tilde{X} = H_\pa (X)$ is from Bob's perspective.
We name this final state $\zeta_{\tilde{X} S F^\comp E \wedge F^\comp = 1} = \Tr_X[\lE_f (\sigma_{X S^\Theta S^{H_\ec} F^\comp } \otimes \rho_{S^{H_\pa}})]$ for the function $f \colon (X, H_\pa) \mapsto H_\pa (X)$.
We compare this to the state $\chi_{\tilde{X}} \otimes \zeta_{S F^\comp E \wedge F^\comp = 1}$ where $\chi_{\tilde{X}}$ is the fully mixed state on $\tilde{X}$.

\begin{proposition}
    Let $\epsilon(\nu)$ be as defined in \eqref{eq:epsilon}.
    Then for any $\nu \in (0, \frac{1}{2} - \delta]$ such that $\epsilon(\nu)^2 < \Pr[F^\comp = 1]_\sigma$, we have
    \begin{equation}
        \| \zeta_{\tilde{X} S F^\comp E \wedge F^\comp = 1} - \chi_{\tilde{X}} \otimes \zeta_{S F^\comp E \wedge F^\comp = 1} \|_{\Tr} \leq \frac{1}{2} 2^{- \frac{1}{2}g(\nu)} + 2 \epsilon(\nu),
    \end{equation}
    where $g(\nu) \coloneqq s(1 - h(\delta + \nu)) - n$.
\end{proposition}
\begin{proof}
    By~\cref{prop:maxent}, we see that
    \begin{equation}
        H_{\max} ^{\epsilon(\nu)} (Z \wedge F^\comp = 1 \mid Y)_\sigma \leq s \cdot h(\delta + \nu).
    \end{equation}
    Together, with~\cref{prop:appunc}, and taking $q = 1 - h(\delta + \nu)$, we get:
    \begin{equation}
        H_{\min} ^\epsilon (X \wedge F^\comp = 1 | VWS^\Theta B')_\sigma \geq sq.
    \end{equation}
    Finally, applying~\cref{prop:lhl}, we obtain the desired inequality.
\end{proof}

For the case where $\epsilon(\nu)^2 \geq \Pr[F^\comp = 1]_\sigma$, we note that the trace distance~$\| \zeta_{\tilde{X} S F^\comp E \wedge F^\comp = 1} - \chi_{\tilde{X}} \otimes \zeta_{S F^\comp E \wedge F^\comp = 1} \|_{\Tr}$ is upper bounded by $\Pr[F^\comp = 1]_\zeta$.
Hence, considering the inequality $\Pr[F^\comp = 1]_\zeta \leq \epsilon(\nu)^2 \leq \epsilon(\nu)$ results in the proof of the following corollary.

\begin{corollary}
    For any $\nu \in (0, \frac{1}{2} - \delta ]$, the following holds:
    \begin{align}
        \begin{split}
            \| \zeta_{\tilde{X} S F^\comp E \wedge F^\comp = 1} - \chi_{\tilde{X}} \otimes \zeta_{S F^\comp E \wedge F^\comp = 1} \|_{\Tr} \leq \frac{1}{2}\sqrt{2^{-s(1 - h(\delta + \nu)) + n}} + 2 \epsilon(\nu).
        \end{split}
    \end{align}
\end{corollary}

Finally, we would like to translate this into a statement about $\abs{p_0 - p_1}$ in~\cref{game:epr}.

\begin{corollary}\label{cor:epr}
    The difference of probabilities 
      \begin{align}
    \vert \Pr[b' = 1 \wedge ok = 1 \mid \cref{game:epr}(0)] - \Pr[b' = 1 \wedge ok = 1 \mid \cref{game:epr}(1)] \vert
    \end{align}
     is negligible.
\end{corollary}
\begin{proof}
    Let $\zeta^b _{\tilde{X} S F^\comp E \wedge F^\comp = 1}$ be the state of $\zeta_{\tilde{X} S F^\comp E \wedge F^\comp = 1}$ in the case that $b \in \{0,1\}$ was selected at the beginning of~\cref{game:epr}.
    Note that the following trace distance is bounded above by a negligible function:
    \begin{align}
        \| \zeta^0 _{\tilde{X} S F^\comp E \wedge F^\comp = 1} - \zeta^1 &_{\tilde{X} S F^\comp E \wedge F^\comp = 1} \|_{\Tr} \\
        &\begin{aligned}
            \leq \| \zeta^0 _{\tilde{X} S F^\comp E \wedge F^\comp = 1} - \chi_{\tilde{X}} \otimes \zeta_{S F^\comp E \wedge F^\comp = 1} \|_{\Tr} \\
             + \| \zeta^1 _{\tilde{X} S F^\comp E \wedge F^\comp = 1} - \chi_{\tilde{X}} \otimes \zeta_{S F^\comp E \wedge F^\comp = 1} \|_{\Tr} \\
        \end{aligned}\\
            &\leq 2\left( \frac{1}{2}\sqrt{2^{-s(1 - h(\delta + \nu)) + n}} + 2 \epsilon(\nu) \right).
    \end{align}
    Next, note the following equality:
    \begin{align}
        &\Pr[b' = 1 \wedge ok = 1 \mid \cref{game:epr}(b)] \\
        &= \sum_{\zeta} \Tr[{\zeta _{\tilde{X} S F^\comp E \wedge F^\comp = 1}}] \Pr[b' = 1 \mid \cref{game:epr}(b)]
    \end{align}
    Hence,
    \begin{align}
        &\vert \Pr[b' = 1 \wedge ok = 1 \mid \cref{game:epr}(0)] - \Pr[b' = 1 \wedge ok = 1 \mid \cref{game:epr}(1)] \vert\\
        &\leq \sum_{\zeta} \Tr[{\zeta _{\tilde{X} S F^\comp E \wedge F^\comp = 1}}] \|\zeta^0 _{\tilde{X} S F^\comp E \wedge F^\comp = 1} - \zeta^1 _{\tilde{X} S F^\comp E \wedge F^\comp = 1} \|_{\Tr}\\
        &\leq \sum_{\zeta} 2 \Tr[{\zeta _{\tilde{X} S F^\comp E \wedge F^\comp = 1}}]  \left( \frac{1}{2}\sqrt{2^{-s(1 - h(\delta + \nu)) + n}} + 2 \epsilon(\nu) \right)\\
        &=2\left( \frac{1}{2}\sqrt{2^{-s(1 - h(\delta + \nu)) + n}} + 2 \epsilon(\nu) \right).
     \end{align}
    The conclusion follows from convexity and the physical interpretation of the trace distance (see \cref{sec:prelim}).
    In particular, the difference in probabilities of obtaining the measurement outcome $b' = 1$ given states $\zeta^0$ and $\zeta^1$ is bounded above by the aforementioned trace distance.
\end{proof}

\subsection{Security Reduction}\label{ss:secred}

We now show that the security of~\cref{game:pm} can be reduced to that of~\cref{game:epr}.
In order to do so, we construct a sequence of games starting at~\cref{game:pm} and ending at~\cref{game:epr}, and show that each transformation can only increase the advantage in distinguishing the case of $b=0$ from the case of $b=1$.

For a game $G$, let $\mathsf{Adv}(G) = \abs{p_0 - p_1}$ be the \emph{advantage}, as defined in \Cref{eq:cdsec}.

\begin{proposition}\label{prop:pm}
    \begin{equation}
        \mathsf{Adv}(\cref{game:pm}) \leq \mathsf{Adv}(\cref{game:epr}). 
    \end{equation}
\end{proposition}
\begin{proof}
    Let $G$ be a game like~\cref{game:pm} except that in~$G$, we run
    \begin{equation}
    \lA_1 (\ketbra{r^\theta}_T \otimes \ketbra{\alpha_1, \alpha_2, \alpha_3}_T  \otimes \rho_S )\,,
    \end{equation} where $\alpha_1, \alpha_2, \alpha_3$ are uniformly random bit strings of the appropriate length. 
    Verification is performed as usual, and if $ok = 1$, we run $\lA_2$ on a state containing $r|_{\bar{\lI}}, \theta, \msg \oplus x \oplus \alpha_1, H_{\ec} (r|_{\lI}) \oplus \alpha_2, \synd(r|_\lI) \oplus \alpha_3, H_{\pa}, H_{\ec}$ along with $\rho'_S \otimes \rho_{T'}$.
    By a change of variable, $\mathsf{Adv}(\cref{game:pm}) = \mathsf{Adv}(G)$.

    Next, we obtain $G'$ from $G$ by defining a new adversary  $\lA'_1$ which is like $\lA'_1$, but only receives part of register $T$. Thus we run
    \begin{equation}
    \lA'_1 (\ketbra{r^\theta}_T  \otimes \rho_S )\,,
    \end{equation}
    and to compensate, we directly give $\lA'_2$ the information that was previously hidden by the $\alpha$ values: we run $\lA'_2$ on a state containing $r|_{\bar{\lI}}, \theta, \msg \oplus x, H_{\ec} (r|_\lI), synd(r|_\lI), H_{\pa}, H_{\ec}$ together with $\rho'_S \otimes \rho_{T'}$.
    Then $\mathsf{Adv}(G) \leq \mathsf{Adv}(G')$, since an adversary $\lA'$ for $G'$ can simulate any adversary~$\lA$ in $G$, and win with the same advantage. To do this, $\lA'$ simply creates its own randomness for $\alpha_1, \alpha_2$ and $\alpha_3$, and adjusts the input to $\lA_2$ based on its own knowledge of  $\msg \oplus x, H_{\ec} (r|_\lI)$ and $\synd(r|_\lI)$.
    Let $G''$ be a game like~$G'$ except that, in $G''$, instead of~$\lA'_1$ being given~$\ketbra{r^\theta}$, $m$ EPR pairs are prepared, yielding quantum systems $A$ and $B$, of which the adversary $\lA'_1$ is given $B$.
    System $A$ is measured in basis $\theta$ yielding a string $r$, and $\lA'_1$ then computes
    \begin{equation}
    \ketbra{y}_D \otimes \rho'_S \otimes \rho_{T'} \leftarrow A'_1 (\rho_B  \otimes \rho_S).
    \end{equation}
    We show that, due to the measurement of system $A$, adversary $\lA'_1$ receives $\ketbra{r^\theta}$, where $r$ is uniformly random. The post-measurement state, conditioned on the measurement of system $A$ yielding outcome $r$, will be equivalent to
    \begin{align}
    \ket{\psi_r} &= \left(\sH ^\theta \ketbra{r} \sH^\theta \otimes 1_m \right) \ket{\EPR^m} \\
    &= \left( \sH^\theta \otimes 1_m \right) \left( \ketbra{r} \otimes 1_m \right) \left( 1_m \otimes \sH^\theta \right) \ket{\EPR^m} \\
    &= \sum_{\tilde r \in \{0, 1\}^m} \frac{1}{2^{m/2}} \left( \sH^\theta \ketbra{r} \ket{\tilde r} \right) \left( \sH^\theta \ket{\tilde r} \right) \\
    &= \frac{1}{2^{m/2}} \left( \sH^\theta \ket{r} \right) \left( \sH^\theta \ket{r} \right) \\
    &= \frac{1}{2^{m/2}} \ket{r^\theta} \otimes \ket{r^\theta},
    \end{align}
    which occurs with probability $\| \ket{\psi_r} \|^2 = \frac{1}{2^{m}}$. Therefore, the advantage in $G'$ is the same as the advantage in $G''$.
Let $G'''$ be a game like $G''$ except that, in $G'''$, instead of system $A$ being measured before running $\lA'_1$, system $A$ is measured after running $\lA'_1$. Then the advantage is unchanged because the measurement and $\lA_1$ act on distinct systems, and therefore commute.
    We note that $G'''$ is like~\cref{game:epr} except that, in the latter game, Bob is the party that prepares the state. Since allowing Bob to select the initial state can only increase the advantage, we get that  $\mathsf{Adv}(G''') \leq \mathsf{Adv}(\cref{game:epr})$. This concludes the proof.
\end{proof}

\begin{theorem} \label{thm:cert-del}
    \cref{sch:pmcd} is certified deletion secure.
\end{theorem}
\begin{proof}
    Through a combination of~\cref{cor:epr} and~\cref{prop:pm}, we arrive at the following inequality:
    \begin{align}
        &\vert \Pr[b' = 1 \wedge ok = 1 \mid \cref{game:epr}(0)] - \Pr[b' = 1 \wedge ok = 1 \mid \cref{game:pm}(1)] \vert\\
        &\leq 2\left( \frac{1}{2}\sqrt{2^{-s(1 - h(\delta + \nu)) + n}} + 2 \epsilon(\nu) \right).
    \end{align}
    Since~\cref{game:pm} is a certified deletion attack for~\cref{sch:pmcd}, we see that~\cref{sch:pmcd} is $\eta$-certified deletion secure for
    \begin{equation}
        \eta(\lambda) = 2\left( \frac{1}{2}\sqrt{2^{-(s(\lambda))(1 - h(\delta + \nu)) + n}} + 2 \exp \left( \frac{-(s(\lambda)) (k(\lambda))^2 \nu^2}{(m(\lambda)) ((k(\lambda))+1)} \right) \right),
    \end{equation}
    which is negligible for large enough functions $s, k$.
\end{proof}